\documentclass{article}

\usepackage{amsmath}
\usepackage{amssymb}
\usepackage{amsthm}
\usepackage{cleveref}

\usepackage{microtype}
\usepackage{caption}
\usepackage[title]{appendix}
\usepackage{tikz}
\usetikzlibrary{decorations.pathreplacing,calligraphy}
\usepackage{geometry}
\usepackage[skip=1.5pt]{parskip}

\usepackage{natbib}

\usepackage{graphicx}
\usepackage{subcaption}

\usepackage{verbatim}
\usepackage{comment}

\usepackage[export]{adjustbox}[2011/08/13]

\usepackage[colorinlistoftodos,disable]{todonotes}

\newcommand{\expyg}{y_{G}}
\newcommand{\expyr}{y_{R}}

\newcommand{\expecpayoff}[1]{\mathbf{E}_{\epsilon_{R}} [\pi_{#1}(\beta; y_G,y_R)]}
\newcommand{\prob}[1]{\mathbb{P}\left[{#1}\right]}
\newcommand{\expec}[2]{\mathbb{E}_{#1}\left[{#2}\right]}

\newcommand{\lbone}{l_{\beta_1}}
\newcommand{\lbtwo}{l_{\beta_2}}

\newcommand{\betaint}{{\beta_{R}}^{-},{\beta_{R}}^{+}}

\crefname{assumption}{assumption}{assumptions}

\crefname{res}{result}{result}
\newtheorem{thm}{Theorem}
\newtheorem{cor}{Corollary}
\newtheorem{prop}{Proposition}
\crefname{prop}{proposition}{propositions}
\newtheorem{lemma}{Lemma}
\crefname{lemma}{lemma}{lemmas}
\newtheorem{defn}{Definition}
\newtheorem{rem}{Remark}

\title{Peace in the Face of Uncertainty: Resource Allocation with Stochastic Armaments}

\author{Sarah Taylor\thanks{%
Taylor: Jesus College, University of Cambridge (email: \texttt{srt44@cam.ac.uk}). We are grateful to  Alastair Langtry for his detailed and insightful comments, and Matt Elliott and Vivek Roy-Chowdhury for their helpful comments. Any remaining errors are the sole responsibility of the authors. This work was supported by the Economic and Social Research Council [award reference ES/J500033/1], and the Keynes Fund.  }}
\date{\today}

\begin{document}
        \maketitle
	\begin{abstract}	
            This paper examines a government's strategic resource allocation choices when facing an opposing group whose military power is uncertain. We investigate how this uncertainty affects the government’s decision to divide resources in a way that either guarantees peace, despite unresolved uncertainty, or risks conflict.  We find that under low uncertainty, the government prefers distributions which ensure peace, while under high uncertainty, they are willing to risk war.  When uncertainty is low, the government's allocation is decreasing in uncertainty. When uncertainty is high it is increasing. The latter leads to an increased probability of fighting and falling total welfare.
        \end{abstract}
        \section{Introduction}\label{sec:introduction}
Ensuring robust and long-term peace is a matter of humanitarian and economic importance.  Indeed, the overwhelming evidence that open conflict hinders development highlights the essential role which successful and stable peace agreements play in facilitating prosperity.\footnote{See \cite{abadie_economic_2003}, \cite{collier_2004_africa}, \cite{cerra_2008_growth}, \cite{suliman_2009_human}, \cite{chamarbagwala_2011_human}, \cite{efobi_2016_terrorism}, and see \cite{blattman_2010_civil} for an overview.  See \cite{collier_economic_1999} for a theoretical foundation.}   The division of resources to prevent war critically hinges on the relative military strengths of the groups involved.  However, there has been limited progress in studying such agreements in settings where at least one group's military strength is uncertain. This is despite numerous studies evidencing that protracted conflict often features sources of uncertainty, such as rough terrain (which makes army sizes difficult to quantify);\footnote{ 'Rough terrain' refers to environments such as dense vegetation or mountainous areas. Numerous studies show that rough terrain is an important indicator of civil conflict: \cite{buhaug_2002_geography}, \cite{FearonJAMESD2003EIaC}, \cite{buhaug_2009_geography}, and \cite{carter_2019_places}.} covert interventions from neighbouring states;\footnote{See \cite{MichalopoulosStelios2016TLEo} and \cite{konig_2017_networks}.} and variable resolutions of the collective-action problem when co-ordinating informally-organised rebel militias.\footnote{\cite{olson_logic_1971} shows that collective action fails in a setting of public good provision.}

The political instability caused by the threat of war can be almost as damaging as war itself - for example through capital flight or the creation of refugees. However, unresolved uncertainty regarding military strength can make it difficult to rule out the possibility of open conflict as a sudden change could render previous agreements moot. We therefore aim to answer two questions: when can peace be guaranteed through the appropriate division of resources despite the presence of unresolved uncertainty,\footnote{We focus on transfers that guarantee peace without outside intervention. This is reasonable as maintaining foreign peace-keeping initiatives is expensive and sometimes unpopular with the intervening country's voters.}  and how does uncertainty impact incentives of regimes to guarantee peace or risk war? 

To this end, we study a simple game wherein a government and a rebel group attempt to agree on dividing up the resources of a country under the threat of war.  
The government proposes a resource transfer to the rebels that will result in peace if accepted. Resources transferred could include land rights, state budget spent on the rebels, or development projects in the area where the rebel group lives. Either party may unilaterally reject the proposed transfer and initiate a war. However, both groups are ex-ante unsure about the probability with which they will win a conflict.\footnote{Note that this differs from the outcome of the conflict being probabilistic, where either group may win with a fixed probability.} We allow the rebels' arms to receive a shock to their armaments after the transfer is proposed, and model uncertainty by the range of possible shocks.

Our work yields results on how uncertain military strength erodes trust and eliminates the existence of peace that is robust to any possible realisation of the shock. We determine the set of transfers that will always ensure peace regardless of how the uncertainty resolves itself, and call these the set of \emph{peace-guaranteeing transfers}. Unsurprisingly we find a critical level of uncertainty beyond which the set of peace-guaranteeing transfers is empty. This bound persists even when both groups possess equal ex-ante armaments and shocks to armaments have an expected value of zero. 

We analyse the government's optimal choice of proposed transfers when uncertainty is low enough that a peace-guaranteeing transfer is always feasible. We find that they do not always prefer peace-guaranteeing transfers. In particular, we find stark differences in the preferences for peace of low- and high-uncertainty regimes. For small levels of uncertainty, the government guarantees peace and proposes transfers which are robust to any shock realisation. This occurs because deterring war creates a surplus through the preventing capital degradation, which the government can then allocate to itself. However, when uncertainty increases above a critical level the government shifts toward risky strategies that may lead to conflict. This occurs because it is unlikely that the rebels will honour the proposed split (if they receive a large positive shock), or be able to enforce it (if they receive a large negative shock). 

Our results therefore crucially classify the level of uncertainty at which regimes can ensure robust peace without outside interventions. A policymaker can trust the government to maintain peace in scenarios where uncertainty is relatively low. However, in high-uncertainty environments, external interventions may be required to ensure peace. 

We then show a surprising non-monotonicity in the government's proposed transfers. When uncertainty is sufficiently low so that the government prefers a peace-guaranteeing transfer, the size of the transfer offered to the rebel group increases in uncertainty. However, once uncertainty is past the threshold for the government to abandon peace-guaranteeing transfers, the size of the transfer offered decreases in uncertainty. Thus, a small amount of uncertainty can ensure a better share of resources for the rebels, whilst a large amount of uncertainty can cause the government to keep most of the resources (in the absence of conflict). 

We also find that, once the government abandons peace-guaranteeing transfers, increasing uncertainty has a material negative impact on key outcomes. In particular, it increases the probability of war and decreases total welfare. These results highlight the importance of guaranteeing peace that is robust to shock outcomes. 

Our work emphasises the importance of policy measures which reduce uncertainty in maintaining peace. Examples of these are the monitoring requirements in the 2015 Iran nuclear deal, Cold War anti-proliferation treaties, and chemical weapons inspections. Our results also explain why conflict is so difficult to eliminate when monitoring is difficult - for example, in wars with numerous actors. However, we also suggest that reducing uncertainty is not valuable once it is below the level at which the government prefers a peace-guaranteeing transfer.

The non-monotonicity in the government's strategy also paints a bleak picture in settings with extreme uncertainty. As uncertainty increases, the government will hoard resources.  This will only be disrupted by rebellions which occur when rebels receive a sufficiently large positive shock to their conflict capacity and challenge the status quo. This prediction matches empirical findings which show that war in the developing world is often precipitated by extreme inequality between groups.\footnote{See \cite{stewart_2008_horizontal}, \cite{cerra_2008_growth}, \cite{cederman_2011_horizontal}.} 


\emph{Literature Review:} Uncertainty under various dimensions has long been used in rationalist explanations for war, which explores the conditions under which peace or conflict emerges (See \cite{fearon_1995_rationalist}, \cite{fearon_2004_some}). Our work contributes by providing another rationalist explanation for why a regime may endogenously prefer war to guaranteed peace but specifically distinguishes between low- and high-uncertainty regimes by showing how the former prefers guaranteed peace whilst the latter opts for riskier strategies which may spark war. This complements the existing literature which has used alternative mechanisms to explain why some agents may prefer open conflict. This includes bounded rationality, information asymmetry (\cite{brito_1985_conflict}, \cite{azam_2003_civil}, \cite{sanchez_2012_bargaining}) and non-destruction of resources in post-war periods (\cite{garfinkel_2000_conflict}).

A large body of literature also addresses the role of commitment problems in generating open conflict.\footnote{ This also features in numerous influential papers in the political science literature, see \cite{powell_2006_war}, \cite{debs_2014_known}.}  Our study shows that governments are more willing to risk war under large uncertainty, as they cannot guarantee stability with a peaceful resource split. \cite{powell_2013_monopolizing} suggests that governments may become more powerful in the absence of conflict, making it difficult for them to commit to refraining from future oppression. \cite{acemoglu_why_2000} assumes that the output of the country (and so the size of future transfers) is uncertain, whilst \cite{acemoglu_2001_theory} assumes the destructiveness of war is stochastic. Our work studies a complementary variation of this where war breaks out due to a change in circumstance rather than because both parties would like to stop fighting but cannot credibly commit to future transfers- although we can reproduce many of the major results of this literature if we view the scale of the shock as parametrising the scale of the commitment problem. This has been widely discussed in the political science literature but not the economics literature.(See \cite{werner_1999_precarious}, \cite{smith_2004_bargaining}, and \cite{werner_2005_making}.)Our work also departs from the previous literature which has only studied how uncertainty influences the ``negotiation gap" - i.e. the interval in which peace-guaranteeing transfers are possible.


The existing work closest to ours is \cite{garfinkel_2021_self}, which investigates ensuring peace-guaranteeing transfers via coalition-proof equilibria in a single-period model where agents simultaneously arm and decide on war.  However, this focuses on arming decisions rather than resource distributions (which are assumed to be fixed). Our work complements this by analysing optimal transfers, which allows us to answer questions about endogenous resource allocation by showing that the behaviour of low- and high-uncertainty regimes differs concerning their preferences for peace. \cite{BesleyTimothy2011TLOP} also studies how the government cannot commit to not using the state purse to fund repression. However, they do not separate arming decisions from fighting decisions and are thus unable to explain what levels of arms differentials would be immune to a commitment problem.\footnote{Belligerents will often choose a large military capacity with no intention of using it. The aim is for the large arms capacity to act as a conflict deterrent rather than to use it in active conflict --- nuclear arms are perhaps the clearest example of this.} 

There has been some work on arming decisions and their resulting impact on resource allocation in the face of uncertainty in other dimensions such as the number of players (\cite{myerson_population_2006}), which subset of the players will compete (\cite{lim_contests_2009}), the cost of conflict (\cite{wasser_incomplete_2013}), and the propagation of shocks on a conflict network (\cite{xu_2022_equilibrium}). Our complements this by studying directly addresses the conditions under which uncertainty armaments impacts resource allocations and in turn leads to war.


Our work therefore contributes to expanding our understanding of how ratio-form contest success functions behave when conflict effort is uncertain. This may be useful beyond its application in this paper as contest success functions are widely used in a variety of settings such as the economics of advertising, tournaments, sports, oligopolistic competition and rent-seeking (See \cite{garfinkel_chapter_2007} for an overview.)

Our paper proceeds as follows. \Cref{sec:model} presents our model and \Cref{sec:results} provides an exposition and discussion of our results. Specifically, \Cref{sec:threshold} studies when peace-guaranteeing transfers exist, and \Cref{sec:opt_beta} studies the government's optimal choice of proposed transfer. \Cref{sec:conclusion} concludes and discusses avenues for future work. The longer proofs have been omitted from the main body of the text and may be found in \Cref{sec:appendix}.

        \section{Model}\label{sec:model}
    	We aim to create a parsimonious two-stage game which models negotiations between two groups under the threat of conflict when there is uncertainty around the relative military strength of the two groups involved. We outline our model below and then discuss its interpretation and some of our modelling decisions.

\subsection{Model Description}
\vspace{1ex}\noindent\emph {Set Up:} We consider a country whose people are split into two groups. We shall call these groups the government, $G$, and the rebels, $R$. The country is equipped with some exogenously determined resources that can be split amongst the two groups for consumption.  We normalise the value of these resources to $1$. Resources allocated to one group for consumption cannot be consumed by the other --- consumption may be split this way because the country is split by ethnic, religious or political divides. 

The government originally controls all the country's resources. They can propose a keep a portion $\beta \in [0,1]$ of the country's resources to themselves, with the remaining $1-\beta$ being transferred to the rebel group. Either group may reject the split and go to war to seize control of all the country's resources. If either group decides to fight then the country descends into war.

Each group is equipped with a pre-determined level of arms, which we shall call $y_G \in \mathbb{R}$ and $y_R \in \mathbb{R}$ for the governments and rebels respectively. The rebels' are subject to an additive shock $\epsilon_R$, where $\epsilon_R$ is drawn from a continuous distribution which has bounded support on the interval $\left[\underline{a},\overline{a}\right] \in \mathbb{R}$.\footnote{As we show below, negotiated peace is impossible when $|\overline{a}-\underline{a}|$ is sufficiently large when preferences are linear. Thus peace would never materialise if the groups were to expect an unbounded shock. However, long-term and robust peace does in reality. It is thus plausible that agents believe shocks are expected to be finite.}

There are no information asymmetries. Both groups know each other's expected conflict efforts, the range of the shock, $[\underline{a},\overline{a}]$, as well as the realisation of the shock $\epsilon_{R}$.\footnote{We may relax this assumption as our results for \Cref{sec:opt_beta} show that the government will always propose a value of $\beta$ so that they will never declare war, regardless of the realisation of $\epsilon_{R}$. Thus, we may produce the same results if we instead assume that the government does not observe the realisation of $\epsilon_{R}$.}

\vspace{1ex}\noindent\emph{Timing: }The timing of the two groups' decisions proceeds as follows:
\begin{enumerate}
	\item The government proposes a transfer $\beta \in [0,1]$;\label{item:split_prop}
	\item The shock $\epsilon_{R}$ is realised;
	\item Each group decides to either accept or reject the proposed split:\label{item:fight_accept}
	\begin{enumerate}
		\item If both groups accept, then each receives their proposed resource share;\footnote{In other words, the government receives $\beta$ and the rebels receive $1-\beta$.} 
		\item If either (or both) of the groups reject, then the country descends into war, and the winner takes all of the state resources.\footnote{This is equivalent to letting the winner choose a new $\mathbf{\beta}$.}
	\end{enumerate}
\end{enumerate}

\vspace{1ex}\noindent\emph{Destructiveness of War:}  We allow for war to be destructive. Thus some of the country's resources are degraded during open conflict by a parameter $\alpha \in [0,1]$. If war breaks out, the resources claimed by the victor at the end of the conflict will be $\alpha$. 

\vspace{1ex} \noindent \emph{Contest Success Function:}
In the event of a conflict, the winner is determined stochastically using a logit-specification Tullock contest success function. We shall denote each group's probability of winning the conflict after the shock has been realised by $p_{i} \in [0,1]$, where $i \in \{G,R\}$:
\begin{equation}\label{eq:pwin}
	p_{G}= \frac{\exp{\expyg}}{\exp{\expyg} + \exp{(\expyr + \epsilon_{R})}} \hspace{3mm}\text{ and } \hspace{3mm} p_{R}= \frac{\exp{\expyr+\epsilon_{R}}}{\exp{\expyg} + \exp{(\expyr + \epsilon_{R})}}.
\end{equation}
The logit specification is widely used in the empirical economics of conflict literature.

\vspace{1ex} \noindent \emph{Utility from fighting:} We assume that utility from consumption is linear. Agents therefore exhibit no risk aversion.

\subsection{Discussion}
Our model allows for the interpretation of $\beta$ to be flexible. It could be a long-term, status-quo distribution of resources, or a proposition made during negotiations take place during an ongoing civil conflict. Our model explains either setting equally well as long as agents view previous conflict and its associated destruction as sunk costs and do not factor them into future conflict decisions.

A key feature of our model is that it is a one-shot game. For this reason, it applies particularly well to settings which feature substantial uncertainty and which lack a functioning democratic process, and so the only mechanism for renegotiation is through violence. For example, over half the countries on the African continent are either considered to be non-democracies or ``hybrid regimes", which feature regular electoral frauds.\footnote{\cite{IDEAS_2021_state}} Regime changes in these countries tend to happen through coups and uprisings.\footnote{Africa has seen 220 coup attempts in the past 50 years, and 45 out of Africa's 54 countries have experienced a coup. (Source: https://projects.voanews.com/african-coups/)}

Our model also applies to settings where comprehensive peace negotiations between belligerents are costly to arrange.\footnote{This is especially true when there are numerous alliances with auxiliary actors --- see \cite{braithwaite_2023_muddying}.} For example, the Nigerian Civil War claimed two million lives and raged for three years without a ceasefire despite international efforts to arrange one. There are several reasons why peace negotiations are so difficult to co-ordinate. Making concessions to an enemy may be associated with reputational costs --- especially if the motivation for war is ideological (\cite{smith_2018_stopping}). Belligerents are often also over-optimistic that a short and decisive victory will materialise(\cite{johnson_2004_overconfidence}). There are numerous examples where even sophisticated military powers have made this mistake: the 1965 US invasion of Vietnam and the 2003 invasion of Iran were initiated with expectations of a decisive win. In a similar vein, the Russian army expected its ongoing invasion of Ukraine to conclude in a matter of days.\footnote{See \cite{slantchev_2003_principle} for a theoretical treatment of informational asymmetries in war.}

We also assume that each group's expected conflict effort is fixed. This is primarily to direct our focus towards the role of uncertainty in relative levels of armaments in determining proposed resource distributions.  Moreover, in many settings, it may be reasonable to assume that belligerents cannot choose their arms flexibly in the short run. Thus, when $\beta$ is proposed, neither group can change their military capacity to improve their bargaining power.

The reason for shocking the rebels' arms is that we imagine that the government controls a standing army. The rebels, however, belong to an informally organised group of allies whose members may be randomly constrained in the capacity to assist in conflict. We could also view the rebels as facing a public goods provision problem when organising their conflict efforts. Thus some agents will want to free-ride instead of participating in the conflict.\footnote{See \cite{olson_logic_1971} for the seminal theoretical discussion of the collective action problem}  This collective action problem could resolve stochastically, resulting in uncertainty about the rebels' achievable level of arms.

Notice that the government may propose a split of resources, they are in no way bound to honour it once the shock is realised. Thus, if they find that the rebels are left sufficiently weak after the realisation of the shock they may be tempted to renege on their proposed split and seize all of the country's resources. Similarly, the rebels cannot commit to accepting even a very favourable transfer as a large positive shock to their conflict effort will tempt them to overturn the proposed split using violence.\footnote{There are many examples of this in history. For example, the Rwandan genocide was began after the Hutu president was killed in a plane crash, thus strenghtening extremist Hutu rhetoric and sparking the war. Similarly, the Bosnian War began shortly after the death of Josip Tito, an autocratic strongman who had kept peace during his reign.}


        \section{Results}\label{sec:results}
    	This section begins with solving for each group's transfer acceptance decision at the second stage of the game after the shock has been resolved. We classify the set of proposed transfers which will be accepted regardless of the realisation of the shock $\epsilon_{R}$. We call this the set of \emph{peace-guaranteeing transfers} and determine at what levels of uncertainty it is non-empty. We then go on to solve for the government's optimal proposed transfer in the first stage of the game --- both when the set of peace-guaranteeing solutions is non-empty and when it is empty.

\subsection{Uncertainty Destroys the Possibility of Robust Peace}\label{sec:threshold}

We start at the second stage of the negotiations where a proposed transfer has been made and the shock has been realized (step \ref{item:fight_accept} in the model). Both groups need to decide whether to accept their original share of resources or attempt to seize all of the country's resources through war. 

Either group is willing to declare war if their expected payoff from doing so exceeds the share of resources they receive from the proposed transfer. Since each group's payoff from fighting depends on their strength relative to their opponents, their best response will depend on the realisation of the shock $\epsilon_{R}$. In particular, the decision to fight is a threshold decision. The government wants to declare war if the realisation of $\epsilon_{R}$ is small, so that the rebels are weak, whilst the rebels want to fight if the realisation of $\epsilon_{R}$ large so that they are strong.

The existence of these thresholds follows from the linearity of utility and the functional form of the contest success function, rather than the distribution of the shock $\epsilon_{R}$ (not just a bounded one). It therefore exists for any shock distribution.

Let the function $BR_i(\epsilon_{R}): [\underline{a},\overline{a}] \rightarrow \{\text{Accept},\text{Fight}\}$ denote group  $i$'s best response function to the realisation of $\epsilon_{R}$. 

\begin{prop}\label{prop:fightthresh}
	Let $\epsilon_{R}$ follow any distribution. The best response for each group is given by:
         \begin{enumerate}
             \item $ BR_{G}(\epsilon_{R}) = \text{Fight} \iff  \beta \leq \alpha   \text{, and } \epsilon_{R} \leq \log(\alpha/\beta -1)+(\expyg - \expyr)$; and
		 \item$BR_{R}(\epsilon_{R}) = \text{Fight} \iff  1-\beta \leq \alpha \text{, and }   \epsilon_{R} \geq \log(\alpha/(1-\beta) -1)+(\expyg - \expyr)$.
         \end{enumerate}
\end{prop}

Notice that the government may set a value of $\beta$ which they find unacceptably low once the shock has resolved, and then go to war over it.

According to \Cref{prop:fightthresh}, once $\beta$ is fixed and the shock is resolved, there are two cases when peace occurs. The first case occurs when war is very destructive ($\alpha<1/2$), and $\beta > \alpha$ and $1-\beta > \alpha$.  The second case is when war is not overly destructive but $\epsilon_{R}$ is realised in a ``Goldilocks" interval of not being too large nor too small. In this case, there is an upper and lower bound on the realisation of the shock for which war does not occur. This leads us to the following remark:
\begin{rem}
	Suppose $\alpha\geq1/2$. Then there is an upper and lower bound on the realisation of $\epsilon_{R}$ for which war will not occur: \begin{equation*}
		\log\left(\frac{\alpha}{\beta} -1\right)+(\expyg - \expyr)	\leq \epsilon_{R} \leq -\log\left(\frac{\alpha}{(1-\beta)}-1\right) + (\expyg - \expyr).
	\end{equation*}
\end{rem}
 Clearly, the range of shocks which allow for peace becomes larger as war becomes more destructive (i.e. $\alpha$ decreases). The effect of increasing the share awarded to the government is ambiguous - this is because increasing the share awarded to either party decreases their willingness to fight, but also increases the willingness of the other group to fight. 

\subsubsection*{Robust Peace is not Always Possible}
\label{sec:peace_gaurantee}
The threat of war can be almost as damaging as war itself. It is therefore useful to determine the set of transfers which will always ensure peace regardless of how the uncertainty resolves itself post-negotiations. We call these the set of \emph{peace-guaranteeing transfers}. 

The existence of such peace-gauranteeing transfers hinges on the fact that $\epsilon_{R}$ is bounded. \Cref{prop:fightthresh} implies that for every realised value of $\varepsilon_{R}$, there exist upper and lower bounds on the values of $\beta$ for which each group will accept it.  Thus, there may be sufficiently small or large values of $\beta$ for which either group's threshold to accept or reject $\beta$ will be satisfied, regardless of the realisation of the shock $\epsilon_{R}$. 

Intuitively, suppose the share of resources awarded to a group is sufficiently large. In that case, they will never want to fight, even when the shock resolves at its extreme so that the group is maximally strong relative to their opponent. The converse also holds for sufficiently small shares of resources. This observation leads us to \Cref{prop:beta_thresholds} below.

\begin{prop}\label{prop:beta_thresholds}
    There are ${\beta_{i}}^{+}$ and ${\beta_{i}}^{-}$  such that, for any $\epsilon_{R}\in [\underline{a},\overline{a}]$ 
    \begin{enumerate}
    		\item $BR_{G}(\epsilon_{R}; \beta)= \text{Accept}$  exactly when $\beta \geq {\beta_G}^{+}$ \label{item:beta_1_star} 
    		\item $BR_{G}(\epsilon_{R}; \beta)=\text{Fight}$  exactly when $\beta < {\beta_G}^{-}$\label{item:beta_1_hat}
    		\item $BR_{R}(\epsilon_{R}; \beta)=\text{Accept}$ exactly when $\beta \leq {\beta_R}^{-}$\label{item:beta_2_star}
    		\item $BR_{R}(\epsilon_{R}; \beta)=\text{Fight}$  exactly when $\beta > {\beta_R}^{+}$. \label{item:beta_2_hat}
    \end{enumerate}
\end{prop}

Notice that \Cref{prop:beta_thresholds} partitions the possible values for $\beta$ into regions where each group's decision is probabilistically determined by the realisation of $\varepsilon_{R}$, and where it is unchanged by $\varepsilon_{R}$. We illustrate this partition of the strategy space in \Cref{fig:beta_part} below. 

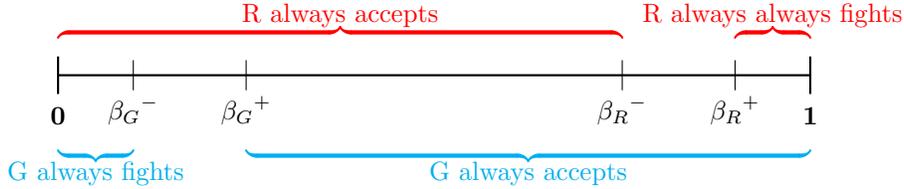
\begin{figure}[h]
	\centering
		\begin{tikzpicture}
			\draw [thick] (0,0) -- (10,0);
			\draw [thick] (0,-0.25) -- (0,0.25);
			\node[label = {[label distance = -9mm]{\textbf{0}}}] at (0,0) {};
			\draw [thick] (10,-0.25) -- (10,0.25);
			\node[label = {[label distance = -9mm]{\textbf{1}}}] at (10,0) {};
			
			\draw(1,-0.2) -- (1,0.2);
			\node[label = {[label distance = -9mm]{${\beta_{G}}^{-}$}}] at (1,0) {};
			\draw [pen colour = {cyan}, ultra thick,decorate, decoration = {calligraphic brace,mirror, amplitude =4pt}] (0,-1) -- (1,-1) node[pos = 0.5, below, cyan]{G always fights};
			
			\draw(9,-0.2) -- (9,0.2);
			\node[label = {[label distance = -9mm]{${\beta_{R}}^{+}$}}] at (9,0) {};
			\draw [pen colour = {red}, ultra thick,decorate, decoration = {calligraphic brace, amplitude =4pt}] (9,0.5) -- (10,0.5) node[pos = 0.5, above, red]{R always always fights};
			
			\draw (2.5,-0.2) -- (2.5,0.2);
			\node[label = {[label distance = -9mm]{${\beta_{G}}^{+}$}}] at (2.5,0) {};
			\draw [pen colour = {cyan}, ultra thick,decorate, decoration = {calligraphic brace,mirror, amplitude =4pt}] (2.5,-1) -- (10,-1) node[pos = 0.5, below, cyan]{G always accepts};
			
			\draw (7.5,-0.2) -- (7.5,0.2);
			\node[label = {[label distance = -9mm]{${\beta_{R}}^{-}$}}] at (7.5,0) {};
			\draw [pen colour = {red}, ultra thick,decorate, decoration = {calligraphic brace, amplitude =4pt}] (0,0.5) -- (7.5,0.5)  node[pos = 0.5, above, red]{R always accepts};
		\end{tikzpicture}
	\caption{An illustration of how the strategy space for $\beta$ is partitioned by fighting outcomes.}\label{fig:beta_part}
\end{figure}

Now, suppose that the parameters in the model are such that ${\beta_{R}}^{-}>{\beta_{G}}^{+}$. Then the interval $[{\beta_{G}}^{+},{\beta_{R}}^{-}]$ forms the set transfers which both groups will definitely accept and so peace after negotiations is guaranteed, regardless of the realisation of the shock $\epsilon_{R}$. 

\begin{rem}
    If $\beta_{G}^{+} > \beta_{R}^{-}$, then there is an interval of $\beta$ which guarantees peace.
\end{rem}

However, it turns out that such an interval does not always exist. If the difference between $\underline{a}$ and $\overline{a}$ is too large, then it may be impossible to find mutually compatible values of $\beta_{G}^{+}$ and $\beta_{R}^{-}$. decreasing as $\overline{a}$ increases. With this in mind, we introduce the following definition of uncertainty:

\begin{defn}
    We define \emph{uncertainty}, denoted $a$, to be given by $ \overline{a}-\underline{a}$. 
\end{defn}

\Cref{thm:peace_threshold} shows if uncertainty is too large, guaranteed peace is impossible.

\begin{thm} \label{thm:peace_threshold}
Suppose $\alpha \geq 1/2$. There exists a critical value of uncertainty, denoted $a_{crit}$, above which there exists no values of $\beta$ that can guarantee peace.
\end{thm}
The intuition behind the above result lies in the comparative statics of $a$ on $\beta{G}^{+}$ and $\beta_{R}^{-}$. This is because as $\underline{a}$ gets smaller, the maximum advantage the government has over the rebels increases, and so  $\beta_{G}^{+}$ increases. The converse holds for $\beta_{R}^{-}$. Thus, an increasing range on $\varepsilon_{R}$ means that the most extreme advantage either group can hope to have over their enemies becomes larger, and so the share of resources they require to not fight is larger. So $a$ increases, $\beta_{G}^{+} \rightarrow \beta_{R}^{-}$ from above until eventually $\beta_{G}^{+}>\beta_{R}^{-}$, and the interval of $\beta$ for which peace is guaranteed shrinks until it eventually disappears.

\Cref{fig:beta_conv} illustrates the intuition of the proof for \Cref{thm:peace_threshold}. For all of the below plots $y_{G}-y_{R}=0$ and $\alpha=0.7$. Here, $t_{G}$ is the upper bound described in \Cref{prop:fightthresh}, and $t_{R}$ is the lower bound described in \Cref{prop:fightthresh}. The region to the right where $\underline{a}$ intersects $t_{G}$ is precisely the set of transfers which the government will accept for all realisations of $\varepsilon_{R}$. The region to the left of where $t_{R}$ intersects $\overline{a}$ is the set of transfers the rebels will never oppose.

\begin{figure}
	\begin{subfigure}[b]{0.49\textwidth}
		\centering
		\includegraphics[width=\textwidth]{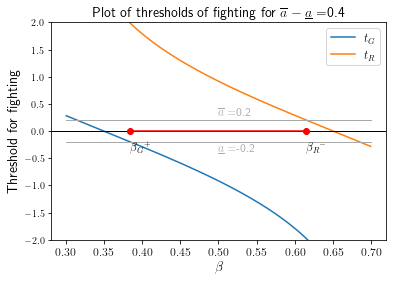}
		\caption{With $a=0.2$, $[{\beta_{G}}^{+},{\beta_{R}}^{-}]$ begins to shrink}\end{subfigure}
	\hfill
	\begin{subfigure}[b]{0.49\textwidth}
		\centering
		\includegraphics[width=\textwidth]{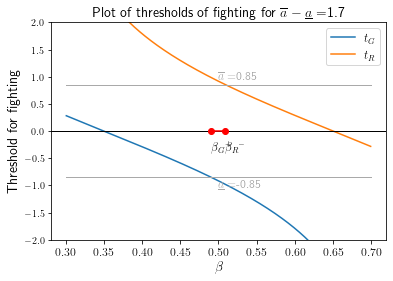}
		\caption{With $a=0.75$, $[{\beta_{G}}^{+},{\beta_{R}}^{-}]$ begins to shrink}
	\end{subfigure}
	\begin{subfigure}[b]{0.49\textwidth}
		\centering
		\includegraphics[width=\textwidth]{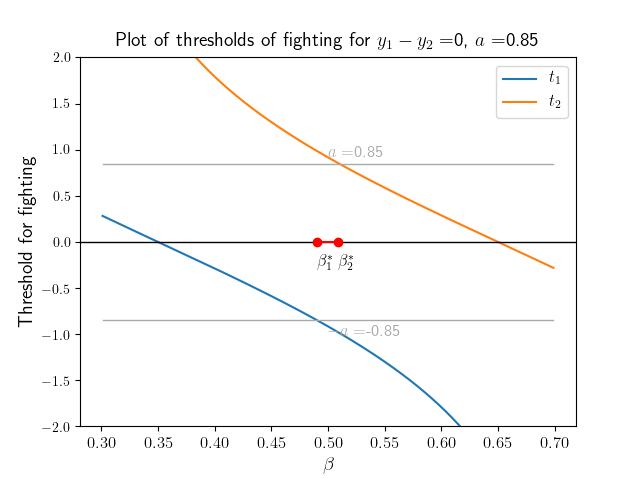}
		\caption{With $a=0.85$, ${\beta_{G}}^{+}\to{\beta_{R}}^{-}$}
	\end{subfigure}
	\hfill
	\begin{subfigure}[b]{0.49\textwidth}
		\centering
		\includegraphics[width=\textwidth]{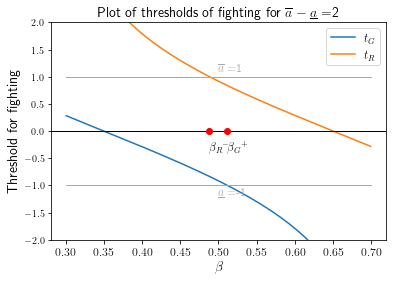}
		\caption{With $a=1$, ${\beta_{G}}^{+}>{\beta_{R}}^{-}$}
	\end{subfigure}
	\caption{As uncertainty increases, the interval of peace-guaranteeing $\beta$ narrows and eventually vanishes. }\label{fig:beta_conv}
\end{figure}

\Cref{thm:peace_threshold} tells us that when uncertainty is low, so that $a<a_{crit}$ it is possible for negotiations to result in lasting peace.  However, the story is more bleak when uncertainty is too large - the uncertainty around the rebels' arms means that committing to peace in this case is impossible as both parties know that, there are certain realisations of the shock for which their odds of winning a war will be high enough that fighting is worthwhile. This is true even when participants have identical ex-ante arms. 

\begin{cor}\label{cor:beta_vals}
  The share of resources awarded to each group from which their decision to accept or fight is certain is:
    \begin{enumerate}
        \item Increasing $\alpha$, and thus decreasing in the destructiveness of war; 
        \item Increasing in their pre-shocked arms advantage over their opponent ($|y_{G}-y_{R}|$).
    \end{enumerate}
    Moreover, the government's shares are decreasing in $\underline{a},\overline{a}$, whilst the rebels' shares are increasing.  
\end{cor} 
These comparative statics make intuitive sense. For the first statement, it is clear that as war becomes increasingly less destructive, fighting becomes an increasingly rewarding exercise, so both parties need a larger share of resources to deter fighting. The second two statements follow from the fact that stronger groups (whether before or after the realisation of the shock) need a larger share of resources as the probability that they will win a conflict is higher.

\Cref{thm:peace_threshold} together with \Cref{prop:beta_thresholds} also implies that when uncertainty is low declarations of war are always one-sided. Specifically, when $a < a_{crit}$ there is no value of $\beta$ at which both groups are willing to declare war after the shock has been realised. This is because we always have that ${\beta_{G}}^{-} \leq {\beta_{G}}^{+}$, and similarly ${\beta_{R}}^{-} \leq {\beta_{R}}^{+}$.  However, the interval wherein the government is willing to declare war is given by $[0, \beta_{G}^{-})$, whilst the rebels are willing to declare war for $({\beta_{R}}^{+},1]$. Since, in that case of $a < a_{crit}$, we have ${\beta_{R}}^{+} > \beta_{G}^{-}$, the intersection of these two intervals is empty.

\Cref{thm:peace_threshold} suggests that resource allocation alone can maintain peace in regions where uncertainty around relative armaments is low. 

We may, in fact, find closed-form solutions for the bounds in \Cref{prop:beta_thresholds}, which we detail in the appendix.\footnote{See \Cref{cor:beta_vals} for further details.} These closed-form solutions readily give rise to the following comparative statics.

When our support for $\epsilon_{R}$ is a symmetric interval we can find a closed-form solution for $a_{crit}$, as shown in \Cref{cor:a_crit_symmetric} below. 
\begin{cor}\label{cor:a_crit_symmetric}
Suppose that the set of possible shocks is symmetric about zero. Then there is a closed-form expression for $a_{crit}$. Moreover, $a_{crit}$ is decreasing in $\alpha$ and increasing in $|y_{G}-y_{R}|$.
\end{cor}
The precise expression of the closed-form solution can be found in the proof of \Cref{cor:a_crit_symmetric} in the appendix. Notice that $a_{crit}$ is minimised when $\expyg = \expyr$, and is increasing in the disparity between $\expyg$ and $\expyr$. This means that unequal armaments can deter war as the weakest party knows they are unlikely to win any conflict, even if they experience the largest positive shock possible. However, we expect to see peace agreements fail when groups are evenly matched in terms of armaments. This makes sense as when both groups are, before shocks, nearly equally likely to win any conflict, they will only be willing to accept resource splits that are close to equal.  The threshold on uncertainty also becomes more restrictive as war becomes increasingly less destructive (as $\alpha \to 1$). This is because increasing $\alpha$ increases the payoff to fighting.

\subsection{The Government's Preferences for Guaranteed Peace vs. Risking War} \label{sec:opt_beta}
 We now aim to understand the government's optimal choice of $\beta$ in the game's first stage. Given its impact on the probability of war (as demonstrated in the previous section), the government faces a trade-off in its choice of $\beta$. As $\beta$ increases, the share of resources that the government gets to keep in peacetime increases.  However, it also increases the risk of war breaking out, rendering the proposed $\beta$ irrelevant. 

For this section, we restrict our attention to when $a<a_{crit}$. This is because our focus is on when a regime will prefer robust peace and when it will risk war, it makes sense to study the portion of the parameter space for which this happens. To obtain clean results we now also assume that the shock $\epsilon_{R}$ is uniformly distributed on the symmetric interval $[-\Tilde{a},\Tilde{a}]$, for some $\Tilde{a} \in \mathbb{R}^{+}$. 

Let $\beta^{*}$ denote an optimal choice of $\beta$ for the government at the game's first stage. It should be immediately clear that if $\beta$ were constrained to be in the set of peace-guaranteeing transfers, then the government would choose $\beta = {\beta_{R}}^{-}$, as this is the maximum value if $\beta$ which guarantees peace. So it remains to determine if the government will choose any other value for $\beta^{*}$.

\begin{prop}\label{prop:beta_unique}
    Suppose that $a<a_{crit}$. Then $\beta^{*}$ lies on the interval $\left[\beta_{R}^{-},\beta_{R}^{+}\right)$ and is generically unique.
\end{prop}

There are two key takeaways from \Cref{prop:beta_unique}. The first is that the government chooses a value of $\beta^{*}$ with which they will always be satisfied, and therefore never declare war when allowed to choose $\beta^{*}$, regardless of shock realisation. In other words $\beta^{*} > \beta_{G}^{+}$, and so $BR_{G}(\epsilon_{R};\beta^{*}) = \text{Accept}$ for all $\epsilon \in [-a,a]$.  The intuition for this is that choosing a value of $\beta$ that is either too high for the rebels always to accept, (i.e. $\beta>\beta_{R}^{-}$), or too low for the government always to accept (i.e. $\beta<\beta_{R}^{+})$,  increases the likelihood of war symmetrically. In the former case, the rebels are unlikely to accept $\beta^{*}$ after the shock realisation, whilst in the latter the government is unlikely to accept it.  However, when a higher value of $\beta$ is chosen the government gets to keep this larger share of resources if $\epsilon_{R}$ is small and the rebels are too weak to reject the proposed transfer. Thus the government is incentivised to award themselves as high a value of $\beta$ as possible.

The second takeaway is that the government never wants to set $\beta$ so high that the rebels will never accept it. In other words, $\beta^{*}<\beta_{R}^{-}$. The intuition for this is that if the rebels reject the proposed value of $\beta$ with certainty, then the choice of $\beta$ becomes irrelevant, as resources will be allocated according to which party succeeds in the conflict. The government would prefer to choose a value of $\beta$ which is slightly below $\beta_{R}^{+}$ which is accepted at least some of the time (specifically when the realisation of $\epsilon_{R}$ is sufficiently small). 

\Cref{prop:beta_unique} is also useful because it shows our game has predictive power: the equilibrium is generically unique. The reason for the qualifier on uniqueness is that for certain values of $\alpha$ and $y_{G}-y_{R}$ there is a single value of $a$ at which the government is indifferent between $\beta_{R}^{-}$ and the interior local maximum, which is strictly larger than $\beta_{R}^{-}$. At this point there is a jump-discontinuity in $\beta_{*}$- the government will suddenly switch from guaranteeing peace with small values of $\beta^{*}$ to demanding increasing large values of $\beta^{*}$. See \Cref{fig:beta_simulations_2} below for an illustration of this phenomenon.
 
This discussion leads us to our next result, which shows that the government undergoes a strategy switch as uncertainty increases. For low levels of uncertainty, the government's optimal strategy is to choose a peace-guaranteeing transfer. However, as $a$ increases, it switches to choosing higher values of $\beta$ which risk war. 

\begin{thm}\label{thm:non_mono}
    Suppose $a<a_{crit}$. The government's choice of $\beta^{*}$  demonstrates a non-monotonicity in $a$, and switches from guaranteeing peace to risking war. In particular, when a is sufficiently small, then  $\beta^{*} = \beta_{R}^{-}$, which is decreasing in $a$. When $a$ is sufficiently large then $\beta^{*}>\beta_{R}^{-}$, and is increasing in $a$.
\end{thm}

Our result is important because it characterises the strategic preferences of high and low uncertainty regimes and shows they are very different. Crucially we show when a regime will switch from guaranteeing peace to risking war, \emph{even when guaranteed peace is still possible.} Moreover, our non-monotonicity result shows that this switch is accompanied by a change in how uncertainty dynamics influence pre-conflict shares of resources. 

Specifically, we find that the government always chooses to guarantee peace as long as uncertainty is not so large as to make it too painful for them.  The fact that the government prefers to guarantee peace here makes intuitive sense. Because war degrades capital, ensuring war does not occur creates some surplus that the government can allocate themselves.\footnote{Formally, suppose that the rebels' probability of winning when they receive the largest possible shock, $\Tilde{a}$, is $p(\Tilde{a})$. Then an offer of $\alpha p(\Tilde{a})$ makes the revels indifferent between fighting and so the government may allocate $1-\alpha p(\Tilde{a}) > \alpha(1-p(\Tilde{a})$ to themselves.} 

In this setting, the government does not need to promise as much to the rebel group to ensure they will never fight, as the likelihood of rebels getting a large positive shock to their arms is small, meaning that $\beta_{R}^{-}$ can be relatively large.  Therefore, for small increases in uncertainty, the government is willing to accept a slightly lower share of resources to ensure peace, as this generates surplus by preventing capital decay, which they can allocate to themselves. Thus, a small increase in uncertainty actually benefits the rebels, as they receive a larger share of the country's resources without having to face open conflict. 

However, because $\beta_{R}^{-}$ is decreasing in $a$, guaranteeing peace becomes increasingly painful for the government as uncertainty grows. This is because the largest possible level of post-shock armaments available to the rebels increases: and thus the rebels require increasingly large shares of the resources to ensure that they do not fight with certainty. At some point, the loss in the government's peace-time shares of resources outweighs the additional surplus gained in avoiding conflict. Thus, the government's strategy switches from guaranteeing peace to risking war (i.e. $\beta^{*} > {\beta_{R}}^{-}$). 

When this is the case, we know that $\beta^{*}$ is increasing in $a$. This may initially seem counter-intuitive, as increasing $\beta$ also increases the probability of war. However, because the rebels are equally likely to receive a large negative shock as they are a large positive shock, there are many shock realisations where the government will retain their large share of $\beta$ as the rebels will be too weak to declare war. Thus appeasing the rebels with a generous transfer which does not fully guarantee peace is less worthwhile as for most of the shock realisation space, the rebels will be satisfied with a smaller transfer.

An implication of the non-monotonicity in  \Cref{thm:non_mono} is that in regimes where uncertainty is high, we expect to see the rebels declare for lower realisations of shocks as the pre-conflict share of resources which is allocated to them is increasingly poor.

We have simulated the values for $\beta^{*}$, and include some illustrations which demonstrate some of the key properties described in our discussion above. See \Cref{fig:beta_simulations_1} illustrates both how $\beta^{*}$ demonstrates the non-monotonicity described in the result above, and how it changes with increasing $\alpha$. Notice that it is decreasing in $\alpha$. This follows from the fact that the rebels need a larger incentive to not fight when war is less destructive. \Cref{fig:beta_simulations_2} illustrates the jump discontinuity in $\beta^{*}$ which occurs for certain parts of the parameter space.

\begin{figure}
        \centering
        \includegraphics[width =0.7\textwidth]{ 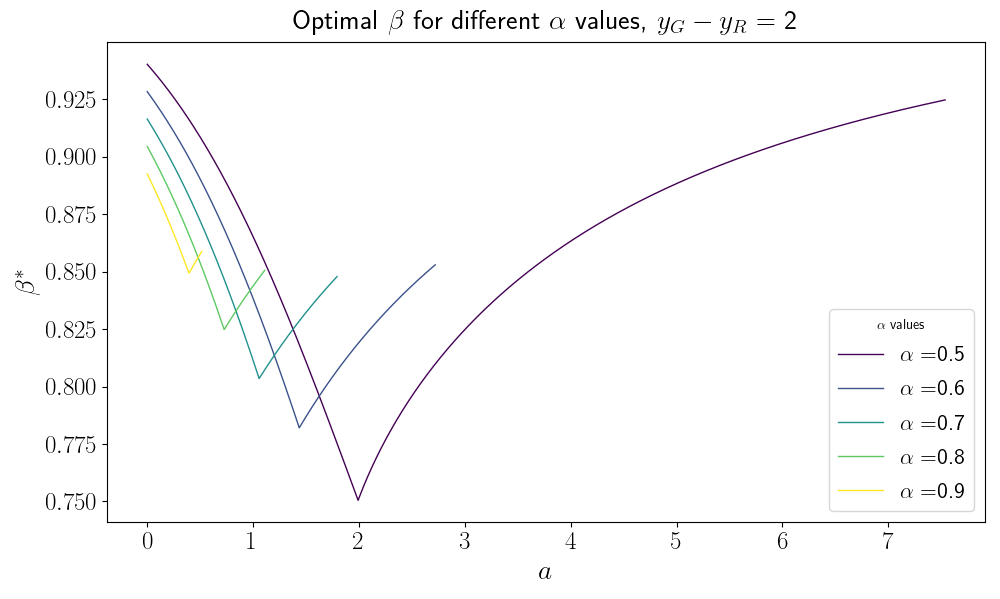} 
        \caption{Simulations illustrating the non-monotonicity in $\beta^{*}$ for increasing $a$, for different values of $\alpha$. The truncating curve length is due to the fact that $a_{crit}$ shrinks with $\alpha$}
        \label{fig:beta_simulations_1}
\end{figure}
\begin{figure}
    \centering
        \includegraphics[width = 0.7\textwidth]{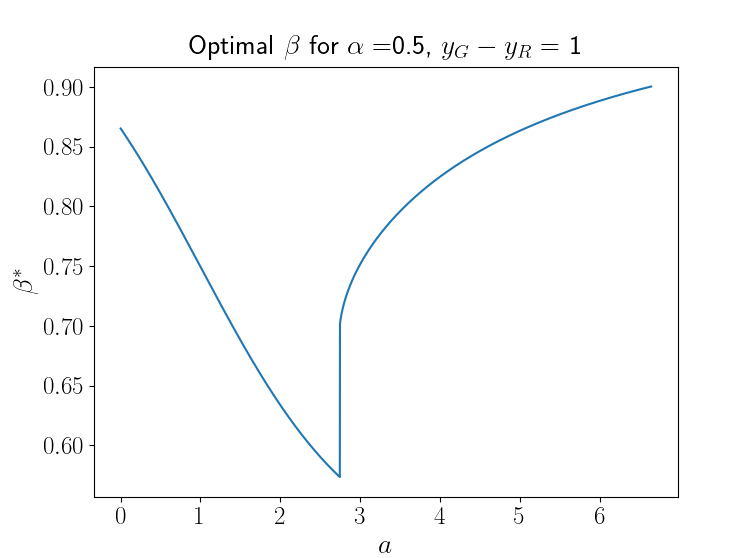}
        \caption{Simulation illustrating the jump discontinuity in $\beta^{*}$ that can occur for some values of the parameter space.}
         \label{fig:beta_simulations_2}
\end{figure}
   
\subsubsection{The Impact of Increasing Uncertainty when Peace is not Guaranteed}
Our results above characterise when a regime switches from guaranteeing peace to risking war. A key question that arises from this is what happens to certain outcomes once this switch has occurred. In the section below, we therefore study how increasing uncertainty impacts key outcomes such as the probability of war and total welfare.

\begin{prop} \label{prop:prob_war}
Suppose $a<a_{crit}$. If $a$ is sufficiently large so that $\beta^{*}>\beta_{R}^{-}$, then the probability of war is in increasing in $a$.
\end{prop}

Our first result shows that once the government risks war,\footnote{Obviously for $\beta^{*}= \beta_{R}$ fighting is not a possibility.} their country is more likely to fight as uncertainty increases. This result is largely driven by the fact that $\beta^{*}$ is increasing $a$ and thus the rebels are willing to fight for lower realisations of $\epsilon_{R}$, all whilst the highest possible realisation of $\epsilon_{R}$ is increasing. Abandoning peace guaranteeing transfers does not simply make war a theoretical possibility but also increases the probability that it occurs. \Cref{fig:prob_war} below shows the results from a simulation, which shows how increasing uncertainty can lead to a steep increase in the probability of conflict occurring. As can be seen in the figure, the probability of fighting increases steeply in $a$ once $\beta^{*}$ switches from $\beta_{R}^{-}$. Thus, once war is on the table, a small increase in uncertainty can dramatically decrease the probability of peace. 
\begin{figure}[h]
    \centering
    \includegraphics[width=0.7\textwidth]{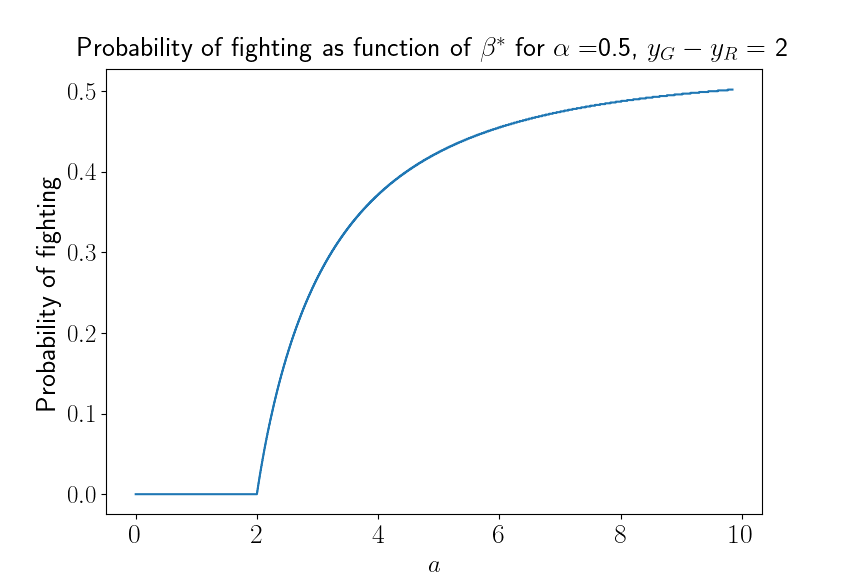}
    \caption{Simulation of how the probability of war changes with $a$. The flat portion of the curve is when $\beta^{*} = \beta_{R}^{-}$, whilst the monotonically increasing portion is when $\beta^{*} > \beta_{R}^{-}$.}
    \label{fig:prob_war}
\end{figure}

\begin{prop}\label{prop:welfare_decreasing}
    When $\beta^{*}=\beta_{R}^{-}$, total welfare is unchanging in $a$. When $\beta^{*}>\beta_{R}^{-}$, total welfare is decreasing in $a$.
\end{prop}

Our second result highlights how, once the government switches away from peace-guaranteeing solutions, increases in uncertainty decrease total welfare in the country. This decrease is entirely driven by the fact that war, and hence the erosion of capital, becomes increasingly likely. Interestingly, when $\beta^{*} = \beta_{R}^{-}$, then the government's welfare strictly decreases in uncertainty. Conversely, in this scenario, the rebels' welfare is increasing when peace-guaranteeing solutions are preferred. This is because $\beta_{R}^{-}$ is decreasing in $a$, as per \Cref{cor:beta_vals}.

Unfortunately, it is not possible to analytically determine the impact of increasing uncertainty on each individual group's welfare when $\beta^{*}>\beta_{R}^{-}$. However, our numerical simulations suggest that, in this case, the government's welfare continues to decrease in $a$ monotonically. The rebels' payoff demonstrates several peculiar non-monotonocities. It is increasing in $a$ when the government prefers $\beta_{R}^{-}$, as predicted by our comparative statics result. However, once the government abandons guaranteed peace (which occurs at the sharp point in \Cref{fig:rebel_payoff}), the rebels' payoff decreases for some range of $a$. However, as $a$ continues to grow, it increases again in $a$- although it remains below the maximum payoff they achieve just before the government abandons $\beta_{R}^{-}$. We illustrate the results of these simulations in \Cref{fig:payoffs} below.

\begin{figure}
    \centering
    \begin{subfigure}{0.49\textwidth}
        \centering
        \includegraphics[width=\textwidth]{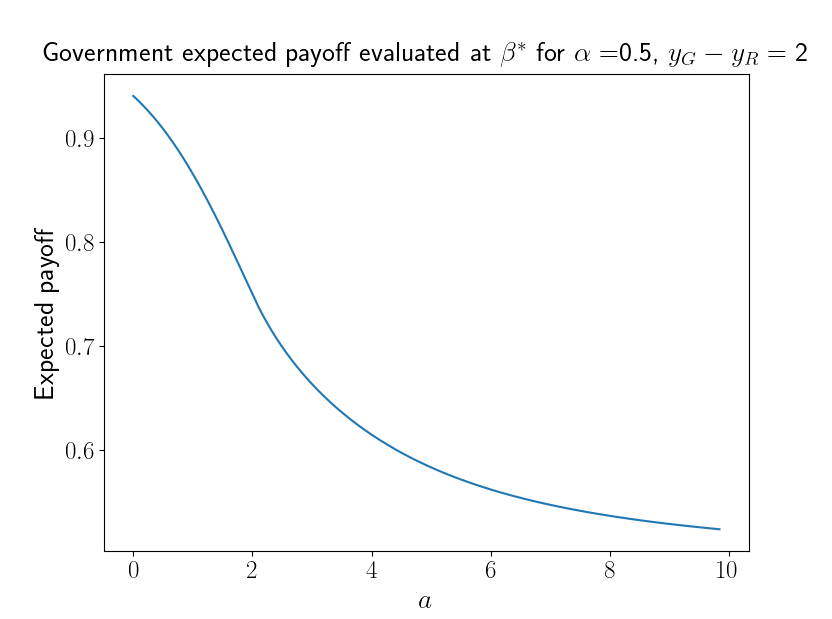}
        \caption{Simulation of the government's expected payoff. The change away from $\beta_{R}^{-}$ occurs at the point where the curvature changes.}
    \end{subfigure}
    \begin{subfigure}{0.49\textwidth}
        \centering
        \includegraphics[width=\textwidth]{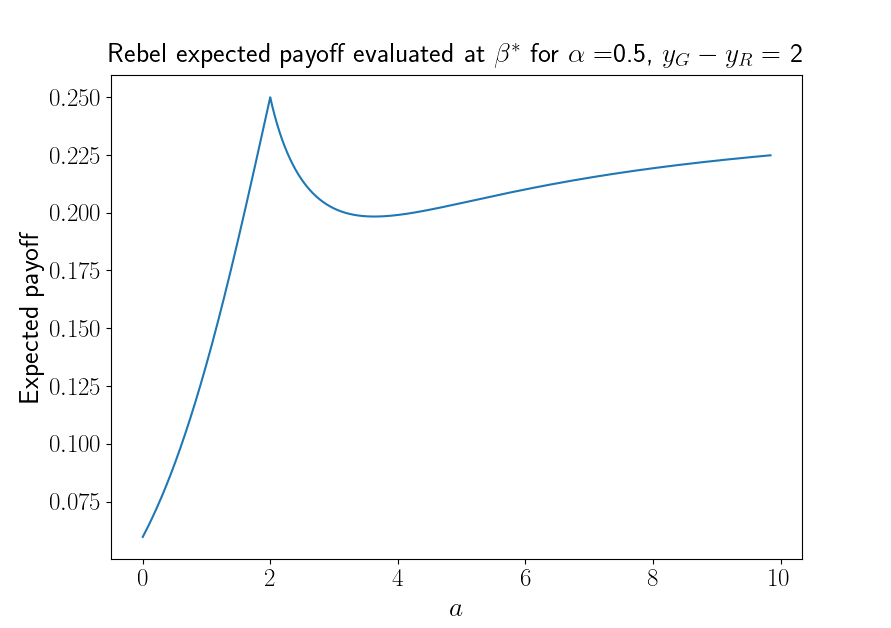}
        \caption{Simulation of the rebels' expected payoff. The change away from $\beta_{R}^{-}$ occurs at the point where of non-differentiability.}
        \label{fig:rebel_payoff}
    \end{subfigure}
    \caption{Expected payoffs for each group evaluated at $\beta^{*}$.}
    \label{fig:payoffs}
\end{figure}

Our results also suggest that institutions with the power to propose resource splits (called the government in our work above) will be co-operative in reducing military uncertainty as their welfare is monotonically decreasing in $a$. Conversely, t the groups whose only means of challenging resource splits is through the threat of violence (called the rebels in our work above) are incentivised to increase uncertainty as this improves their welfare even when war does not break out. 

This may have implications, for example, in nuclear nonproliferation negotiations, where militarily weak countries with low negotiating power have an incentive to increase uncertainty around their nuclear capabilities. Similarly, it may explain why peace deals are so difficult to maintain in conflicts which feature terrorism or guerilla warfare, where the strength of the rebels is difficult to monitor and so governments are not incentivised to propose peace-guaranteeing transfers.

\subsubsection{Policy Implications}
Our results above have shown that increasing uncertainty has material negative effects once peace-guaranteeing transfers are abandoned. This highlights the importance of policy interventions such as monitoring and transparent militarisation which reduce uncertainty to below the level at which this abandonment occurs. 

The takeaway from our results about the role of the government in maintaining peace is somewhat mixed. On the one hand, it shows that a peace-guaranteeing transfer is sometimes aligned with the government's optimal proposed transfer. Thus, in these scenarios, if preventing war and state repression is a policymaker's primary objective then the government can be trusted to ensure peace with minimal outside intervention. However, on the other hand, \Cref{thm:non_mono} also shows that the governments operating in high uncertainty environments may rationally choose to risk war in the absence of intervention. This is especially a problem if we want to prevent war for non-resource-distribution reasons such as preventing loss of human life.

Avoiding war is valuable, so preventing regimes from switching to a scenario where they risk war is important.  This is especially true given our finding that the move away from peace-guaranteeing transfers does not simply put an infinitesimally small likelihood of war on the table, but can lead to substantial increases in the likelihood of war. This goes beyond our discussion in the section above wherein substantial uncertainty simply eliminates the set of peace-guaranteeing transfers.  Thus, if we are aiming to prevent the loss of human life through open conflict we need some intervention in negotiations to ensure the selection of a peace-guaranteeing solution: ideally, one that reduces uncertainty in some way. 

However, the results above also show that reducing uncertainty is not valuable once it is below the level at which the government prefers a peace-guaranteeing transfer. This is because the results above show that increasing uncertainty has a starkly different impact depending on the regime. When $\beta^{*} = \beta_{R}^{-}$, increasing uncertainty has no impact on the likelihood of war nor total welfare. However, once uncertainty. Therefore, interventions should take the comparative statics above into account.  If uncertainty is above the critical level at which peace-guaranteeing transfers are preferred, then a small increase in uncertainty leads to a less favourable share of pre-conflict resources for the rebels.

Our results also emphasize the importance of efforts to enable parties to renegotiate their resource shares without returning to violence when there is unresolved uncertainty. Part of what drives open conflict in our model is that neither party may renegotiate after the shock is realised. Although there are many instances where returning to the negotiation table is difficult (see the model section for further discussion), our work emphasises that efforts to improve negotiations are essential.

        \section{Conclusion}\label{sec:conclusion}
    	We have produced a parsimonious one-shot model which attempts to understand how uncertainty regarding achievable conflict effort impacts the effectiveness of negotiated outcomes in ensuring peace. Our findings show that when uncertainty is too large, it is impossible to obtain a split of resources which guarantees peace, and thus open conflict is always on the table. We then study how this influences the government's optimal choice of the proposed transfer and find that when uncertainty is sufficiently low, the government chooses a transfer that overcomes the commitment problem and ensures lasting peace. However, as uncertainty increases the government abandons this strategy and instead opts for an increasing share of resources to themselves. This happens even if the set of peace-guaranteeing transfers is non-empty. 

These findings are interesting for several reasons. Firstly, they show that conflict is often unpreventable, even in the face of perfect information and symmetric levels of armaments. Our results may explain why some conflict scenarios are so difficult to resolve. Secondly, they provide insights about when outsider intervention is necessary during civil conflicts to ensure lasting peace: in particular, if uncertainty is sufficiently low then governments can be trusted to choose a set of transfers which ensure peace. Finally, our results may explain why powerful states are willing to offer large concessions to weak but unpredictable opponents. 

Further extensions may include allowing for a more general expression distribution of shocks, and utilising the power-form specification of the Tullock contest success function (as opposed to the logit-form which is used here).  
        \newpage
        \appendix
	\section{Appendix}\label{sec:appendix}
	    \subsection{Proofs for \Cref{sec:threshold}}
\subsubsection{Proof for \Cref{prop:fightthresh}}
\begin{proof}
	We have that group $i$ will want to fight if and only if their expected payoff from  fighting is greater than the payoff they would get from accepting the proposed split:
	\begin{align*}
		p_{i}\left(\expyg,\expyr,\epsilon_{R}\right)\alpha \geq \beta_i.
	\end{align*}
	
	Since $p_{i}\left(\expyg,\expyr,\epsilon_{R}\right) \leq 1$, we have that it is never optimal to declare war if $\alpha < \beta_{i}$. Thus the government will never declare war if $\beta > \alpha$, and similarly the rebels will never declare war if $1-\beta>\alpha$.
	
	Now assume that $\beta<\alpha$. The government will want to fight if
	\begin{align*}
		\frac{\exp(\expyg)}{\exp(\expyg)+\exp(\expyr+\epsilon_{R})}\alpha & \geq \beta,\\
		\iff \exp(y_{R}-y_{G})\exp(\epsilon_{R})  &\leq  \frac{\alpha}{\beta}-1\\
		\iff \epsilon_{R} & \leq (\expyg-\expyr) + \log\left(\frac{\alpha}{\beta}-1\right);
	\end{align*}
	where taking the logarithm is legitimate by the assumption $\beta\leq\alpha$.
	
	Similarly, if we instead assume $1-\beta\leq\alpha$, group $2$ will want to fight if
	\begin{align*}
		\frac{exp(\expyr+\epsilon_{R})}{\exp(\expyg)+\exp(\expyr+\epsilon_{R})}\alpha &\geq 1-\beta,\\
		\iff  \exp(\expyg-\expyr)\exp(-\epsilon_{R})  &\leq \frac{\alpha}{1-\beta}-1\\
		\iff -\epsilon_{R} &\leq (\expyr-\expyg) + \log\left(\frac{\alpha}{1-\beta}-1\right)\\
		\iff \epsilon_{R} &\geq (\expyg-\expyr) -\log\left(\frac{\alpha}{1-\beta}-1\right).
	\end{align*}
\end{proof}

\subsubsection{Proof of \Cref{prop:beta_thresholds}}
The proof of \Cref{prop:beta_thresholds} is below. We first require the following two lemmas.

To aid our discussion, we will denote the respective lower and upper bounds on the shock realisation for which neither group fights as 
\begin{align*}
	{{t}_{G}} &= \log\left(\frac{\alpha}{\beta} -1\right)+(\expyg - \expyr); \text{ and }\\
	{{t}_{R}} &= -\log\left(\frac{\alpha}{(1-\beta)}-1\right) + (\expyg - \expyr).
\end{align*}

\begin{lemma}\label{lemma:l_bone_l_btwo}
	We always have that $t_{R} \geq t_{G}$.
\end{lemma}

\begin{proof}
	For the statement to hold we need that
	\begin{align}
		\log\left(\frac{1-\beta}{\alpha-(1-\beta)}\right)  & \geq \log\left(\frac{\alpha}{\beta} -1\right)\label{cond:orig_alpha}, \\
		\nonumber	\iff \log\left(\frac{1-\beta}{\alpha-(1-\beta)}\right) &\geq \log\left(\frac{\alpha-\beta}{\beta}\right),\\
		\nonumber	\iff (1-\beta)\beta &\geq (\alpha-(1-\beta))(\alpha-\beta);
	\end{align}
	
	where the final line holds because we assume that $\alpha - (1-\beta)>0$ and the fact that $\log$ is a monotonically increasing. Multiplying out the brackets gives: 
	
	\begin{align}
		\nonumber \beta - \beta^{2} &\geq \alpha^2 -\alpha + \alpha\beta -\alpha\beta +\beta-\beta^2\\
		\iff 0 &\geq \alpha^2 -\alpha.\label{cond:alpha_l_zero}
	\end{align}
	But, line \ref{cond:alpha_l_zero} is true for all $\alpha$ $\in [0,1]$, thus the statement in line \ref{cond:orig_alpha} always holds. 
\end{proof}

\begin{lemma}\label{lem:thresh_mono}
	The threshold $t_{G}$ is monotonically decreasing in $\beta$ if $\beta < \alpha$.
	
	The threshold $t_{R}$ is monotonically decreasing in $\beta$ if $\beta > 1-\alpha$
\end{lemma}

\begin{proof}
	This immediately follows from calculating the derivative of each threshold
	\begin{equation*}
		\dfrac{\partial t_{G}}{\partial \beta} = -\frac{\alpha}{\beta(\alpha-\beta)} \hspace{1cm} -\dfrac{\partial t_{R}}{\partial \beta} = -\frac{\alpha}{(1-\beta)(\alpha-(1-\beta))};
	\end{equation*}
	the first of which is negative if $\alpha> \beta$, the second if $\alpha>(1-\beta)$.
\end{proof}

We may now commence with the main proof. 

\begin{proof}
	We begin by proving claims \ref{item:beta_1_star} and \ref{item:beta_1_hat}.  First, notice that $BR_{G}(\epsilon_{R})=1$  regardless of the realisation $\epsilon_{R}$ if and only if $t_{G}<\underline{a}$ - so that the upper bound on the level of shock at which the government never fights is always smaller than the smallest possible shock realisation.  Similarly, $BR_{G}(\epsilon_{R})=0$ $\forall \epsilon_{R}\in [\underline{a},\overline{a}]$ if and only if $t_{G}>\overline{a}$. 
	
	From \Cref{lem:thresh_mono} we know  that $t_{G}$ is monotonically decreasing in $\beta$ on the range $(0,\alpha)$.  Thus if there is some ${\beta_{G}}^{+}$ such that $t_{G}|_{{\beta_{G}}^{+}} = \underline{a}$, then $t_{G} < \underline{a}$ iff $\beta>{\beta_{G}}^{+}$. Similarly, if there is some  ${\beta_{G}}^{-}$ such that $t_{G}|_{{\beta_{G}}^{-}}  = \overline{a}$, then $t_{G}>\overline{a}$ iff $\beta<{\beta_{G}}^{-}$.
	
	Now, notice that $\lim_{\beta \to 0}t_{G} = \infty$, and $\lim_{\beta\to\alpha^{-}}t_{G} = \infty$.  Since $t_{G}$ is clearly continuous in $\beta$, by Rolle's theorem we know that ${\beta_{G}}^{+}$ and ${\beta_{G}}^{-}$ must exist.	By the monotonicity of $t_{G}$, we know that the values of ${\beta_{G}}^{+}$ and ${\beta_{G}}^{-}$ are unique.
	
	Indeed, we may calculate ${\beta_{G}}^{+}$ explicitly:
	\begin{align}
		\nonumber &\log\left(\frac{\alpha-{\beta_{G}}^{+}}{{\beta_{G}}^{+}}\right) + (\expyg-\expyr) = \underline{a}\\
		&\Rightarrow {\beta_{G}}^{+} = \frac{\alpha}{\exp(\expyr-\expyg+\underline{a})+1}\label{eq:beta_1_bound}.
	\end{align}   
	Notice that ${\beta_{G}}^{+}<\alpha$, and since $BR_{G}(\epsilon_{R})= \text{Accept}$ if $\beta>\alpha$, we have that $BR_{G}(\epsilon_{R})=\text{Accept}$ on $[\beta_{1^{*}},\alpha]\cup[\alpha,1] = [{\beta_{G}}^{+},1]$. Thus we have found the interval on which the government always accepts the proposed transfer. Similarly, the government always fights if $\beta \in [0,{\beta_{G}}^{-}]$. 
	
	We may also calculate ${\beta_{G}}^{-}$ explicitly:
	\begin{align}
		\nonumber &\log\left(\frac{\alpha-{\beta_{G}}^{-}}{{\beta_{G}}^{-}}\right) + (\expyg-\expyr) = \overline{a}\\
		&\Rightarrow {\beta_{G}}^{-} = \frac{\alpha}{\exp(\expyr-\expyg+\overline{a})+1}.
	\end{align}
	
	The government will thus never accept $\beta$ if $\beta \in [0,{\beta_{G}}^{-}]$. 
	
	We may use a similar logic to prove claims \ref{item:beta_2_hat} and \ref{item:beta_2_star}. From \Cref{lem:thresh_mono}, we have that $t_{R}$ is monotonically decreasing in $\beta$ on $(1-\alpha,1)$. We have that $BR_{G}(\epsilon_{R})=\text{Accept}$ for all realisations of $\epsilon_{R}$ in $[\underline{a},\overline{a}]$ iff $t_{R}>\overline{a}$.  Similarly, $BR_{G}(\epsilon_{R})=\text{Fight}$ $\forall \epsilon_{R}\in [\underline{a},\overline{a}]$ iff $t_{R}<\underline{a}$. 
	
	As above, the existence of ${\beta_{R}}^{+}$ and ${\beta_{R}}^{-}$ follows from the fact that $\lim_{\beta ^{+}\to (1-\alpha)} t_{R} = \infty$ and $\lim_{\beta \to 1}t_{R} = \infty$ in conjunction with Rolle's theorem. 
	
	We may calculate ${\beta_{R}}^{+}$ and ${\beta_{R}}^{-}$ explicitly:
	\begin{align}
		\nonumber&\log\left(\frac{1-{\beta_{R}}^{+}}{\alpha-(1-{\beta_{R}}^{+})}\right) + (\expyg-\expyr) = \underline{a}\\
		&\Rightarrow {\beta_{R}}^{+} = 1 - \frac{\alpha}{\exp(\expyg-\expyr-\underline{a})+ 1}.\\
		\nonumber&\log\left(\frac{1-{\beta_{R}}^{-}}{\alpha-(1-{\beta_{R}}^{-})}\right) + (\expyg-\expyr) = \overline{a}\\
		&\Rightarrow {\beta_{R}}^{-} = 1 - \frac{\alpha}{\exp(\expyg-\expyr-\overline{a})+ 1} \label{eq:beta_2_bound};
	\end{align}
	Notice that $1-{\beta_{R}}^{-}<\alpha$, and since $BR_{G}(\epsilon_{R})=1$ if $1-\beta<\alpha$, we have that $BR_{G}(\epsilon_{R})=1$ on $[0,1-\alpha]\cup[1-\alpha,{\beta_{R}}^{-}] = [0,{\beta_{R}}^{-}]$. Thus we have found the interval of $\beta$ for which the rebel group never fights. Similarly, the rebel group will always fight if $\beta  \in [{\beta_{R}}^{+},1]$.
\end{proof}

\subsubsection{Proof of \Cref{thm:peace_threshold}}
\begin{proof}
	For notational brevity, denote $x = \expyg-\expyr$. We want to find the threshold on $\overline{a}-\underline{a}$ for which ${\beta_{R}}^{-}\geq {\beta_{G}}^{+}$. Specifically, we are looking for pairs $(\underline{a}$, $\overline{a})$ which satisfy
	
	\begin{equation}\label{eq:a_thresh}
		1-\frac{\alpha}{1+e^{(x-\overline{a})}} \geq \frac{\alpha}{1+e^{(-x+\underline{a})}}.
	\end{equation}
	
	We will have that the value for $a = \underline{a}-\overline{a}$ is unique if ${\beta_{R}}^{-}-{\beta_{G}}^{+}$ is monotonically decreasing in $a$. Indeed, notice that 
	\begin{equation*}
		\frac{\partial {\beta_{R}}^{-}}{\partial a} = -\frac{\alpha}{4\cosh^{2}\left( \frac{a+\underline{a}-\expyg+\expyr}{2}\right)},
	\end{equation*}
	which is always negative. Thus ${\beta_{R}}^{-}$ is monotonically decreasing in $a$. Conversely, we have that 
	\begin{equation*}
		\frac{\partial {\beta_{G}}^{+}}{\partial a} = \frac{\alpha \exp(\expyr-\expyg+a-\overline{a})}{\left(\exp(\expyr-\expyg+a-\overline{a})+1\right)^{2}},
	\end{equation*}
	which is always positive. Thus any root of the difference between the two functions is unique.
	
	So all that remains is to prove that such a root always exists. It is clear that this function is continuous in $\overline{a}-\underline{a}$, and so we may apply Rolle's theorem. Notice that we may factorise \ref{eq:a_thresh} and multiply through by $e^{a}$ to obtain:
	\begin{equation*}
		{\beta_{R}}^{-}-{\beta_{G}}^{+} = (1-2\alpha)e^{a} + (1-\alpha)e^{\frac{a}{2}}\left(e^{x-\left(\frac{\overline{a}+\underline{a}}{2}\right)} + \frac{1}{e^{x-\left(\frac{\overline{a}+\underline{a}}{2}\right)}}\right) +1.
	\end{equation*}
	
	Notice that the RHS of the inequality is a concave-down quadratic polynomial in $e^{\frac{a}{2}}$. Thus, if this function is positive on any part of its domain, it must have two roots between which it is positive.   For brevity, we shall denote the second term as $\eta + 1/\eta$. Thus we have that the roots are given 
	\begin{equation*}
		z_{+,-} = \frac{-(1-\alpha)\left(\eta + \frac{1}{\eta}\right) \pm \sqrt{(1-\alpha)^{2}}\left(\eta + \frac{1}{\eta}\right)^{2} -4(1-2\alpha)}{2(1-2\alpha)}.
	\end{equation*}
	Notice that $\sqrt{(1-\alpha)^{2}\left(\eta + \frac{1}{\eta}\right)^{2} -4(1-2\alpha)} \geq  |-(1-\alpha)\left(\eta + \frac{1}{\eta}\right)|$. Thus this polynomial has exactly one positive and one negative root. We can rule out the negative root as $e^{\frac{a}{2}}$ must always be positive over the real numbers. Notice that because the denominator is negative, we have that $z_{-}$ is the positive root. Thus we have that $	{\beta_{R}}^{-}-{\beta_{G}}^{+}$ is positive if $a$ is such that $e^{a/2}>z_{-}$, and negative otherwise. Since the exponential function is increasing in $a$, we have that there is a maximum value of $a$ for which ${\beta_{R}}^{-}-{\beta_{G}}^{+}$ is positive.
\end{proof}

\begin{cor}\label{cor:beta_vals}
  We have that\todo{Add in small comparative statics result in $\alpha$, max difference between each group's achieved conflict capacity} 
	\begin{align*}
		{\beta_{G}}^{+}&=\frac{\alpha}{\exp(\expyr-\expyg+\underline{a})+1} \hspace{1cm} {\beta_{R}}^{-} =1 - \frac{\alpha}{\exp(\expyg-\expyr-\overline{a})+ 1}\\
		{\beta_{G}}^{-} &= \frac{\alpha}{\exp(\expyr-\expyg+\overline{a})+1} \hspace{1cm} {\beta_{R}}^{+} =  1 - \frac{\alpha}{\exp(\expyg-\expyr-\underline{a})+ 1}.
	\end{align*}
\end{cor}

\subsubsection{Proof of \Cref{cor:a_crit_symmetric}}
The proof of \Cref{cor:a_crit_symmetric} is below.
\begin{proof}
	\emph{Closed form expression for $a_{crit}$:} We now have that ${\beta_{R}}^{-} \geq {\beta_{G}}^{+}$ iff\todo[restructure]{This is probably unnecessary and could be dropped in directly from the previous proof.} 
	\begin{align}
		\nonumber  1-\frac{\alpha}{e^{x-a}+1}  &\geq \frac{\alpha}{e^{-x-a}+1}\\
		\nonumber \left(1+e^{x-a}\right)\left(1+e^{-x-a}\right) - \alpha\left(1+e^{-x-a}\right) &\geq \alpha\left(1+e^{x-a}\right)\\
		\nonumber \iff (1-2\alpha)e^{2a} + (1-\alpha)\left(e^{x}+e^{-x}\right)e^{a} + 1 &\geq 0.		
	\end{align}
	
	Notice that the condition in the final line is simply a quadratic equation in $e^{a}$. Let $z \equiv e^{a}$. We then have that the roots of the quadratic are given by:
	\begin{equation*}
		z_{1,2} = \frac{\zeta \pm \sqrt{\zeta^{2} -4(1-2\alpha)}} {2(1-2\alpha)},
	\end{equation*}
	where $\zeta = (\alpha-1)\left(e^{x}+e^{-x}\right)$. Note that $\zeta<0$ as $\alpha<1$.
	
	Now, since $e^{a}$ is always positive, we need to focus on finding the positive root. Notice that if $(1-2\alpha)>0$ (so that $\alpha>\frac{1}{2}$) then $|\zeta|< \sqrt{\zeta^{2} -4(1-2\alpha)}$, which implies that $z_{1,2}<0$. This means that the roots of the polynomial are both negative and so $(z-z_1)(z-z_2) >0$ $\forall z>0$, as the polynomial is concave up. Thus a peace-guaranteeing solution always exists. 
	
	If $\alpha \geq \frac{1}{2}$ then $1-2(\alpha)\leq 0$, which means that $|\zeta|> \sqrt{\zeta^{2} -4(1-2\alpha)}$. Thus $\zeta +  \sqrt{\zeta^{2} -4(1-2\alpha)} >0$, and $\zeta -  \sqrt{\zeta^{2} -4(1-2\alpha)} <0$. So the positive root is given by 
	\begin{equation*}
		z_2 = \frac{\zeta - \sqrt{\zeta^{2} -4(1-2\alpha)}} {2(1-2\alpha)}.
	\end{equation*}  
	
	Since the polynomial is concave down, we have that $(z-z_{1})(z-z_{2}) >0$ if $z_{1}<z<z_2$. Thus we have that
	\begin{align*}
		(1-2\alpha)e^{2a} + (1-\alpha)\left(e^{x}+e^{-x}\right)e^{a} + 1 &\geq 0\\
		\iff e^{a} &\leq  \frac{\zeta - \sqrt{\zeta^{2} -4(1-2\alpha)}} {2(1-2\alpha)}\\
		\iff a &\leq \log\left(\frac{\zeta - \sqrt{\zeta^{2} -4(1-2\alpha)}} {2(1-2\alpha)}\right),
	\end{align*}
	
	which gives the desired threshold. Thus $z_{2}$ Because the left-hand side is finite, and $a$ may be arbitrarily large, we have that for $a$ sufficiently large, the threshold is crossed.

    \emph{$a_{crit}$ is increasing in $|y_{G}-y_{R}|$:} Differentiating $z_2$ with respect to $|y_G-y_R|$ yields
    \begin{equation}
        \frac{(e^{2 |y_{G}-y_{R}|} -1)(1-\alpha)}{\sqrt{\alpha^{2}( e^{2 |y_{G}-y_{R}|} + 1)^{2} + 2\alpha( e^{2 |y_{G}-y_{R}|} -1)^{2}  + (e^{2 |y_{G}-y_{R}|} -  1)^{2}}},
    \end{equation}
    which is always positive.

    \emph{$a_{crit}$ is decreasing in $\alpha$:} Let $exp(y_{G}-y_{R}) + exp(y_{R}-y_{G}) = x$, and then differentiating $z_2$ with respect to $\alpha$ yields:
    \begin{align*}
        &\frac{  x^{2} \cdot \left(1 - \alpha\right) + x \sqrt{8 \alpha + x^{2} \left(\alpha - 1\right)^{2} - 4} - 4}{\left(x \left(\alpha - 1\right) - \sqrt{8 \alpha + x^{2} \left(\alpha - 1\right)^{2} - 4}\right) \sqrt{8 \alpha + x^{2} \left(\alpha - 1\right)^{2} - 4}} -\frac{2}{2\alpha-1}.
    \end{align*}
The denominator second term is negative as $\left(2 \alpha - 1\right)>0$ as long as $\alpha>1/2$. Moreover, the denominator of the first term is negative as $x(\alpha-1)<0$, as $x>0$ and $\alpha<1$. Now, notice that the minimum value of $x$ is 2, so that 
\begin{align*}
     &x^{2} \cdot \left(1 - \alpha\right) + x \sqrt{8 \alpha + x^{2} \left(\alpha - 1\right)^{2} - 4} - 4\\
     &\geq 4(1-\alpha)+2\sqrt{8\alpha + 4(\alpha-1)-4} -4
     =3\alpha \geq 0,
\end{align*}
so that the second term is also negative. (provided $\alpha\geq 2$).
\end{proof}

\subsection{Proofs for \Cref{sec:opt_beta}}
Let $\tilde{p}_{G}$ be the expected probability that the government wins the conflict, taken over all possible realisations of the shock:
\begin{equation*}
	\tilde{p}_{G} = \mathbb{E}_{\epsilon_{R}}[p_{G}]= \int_{-a}^{a} \frac{
		e^{\expyg}}{e^{\expyg}+e^{\expyr+\epsilon_{R}}}\mathrm{d}\epsilon_{R} = \frac{\alpha}{2a} \frac{\log\left(e^{\expyg-\expyr+a}+1\right)}{\log\left(e^{\expyg-\expyr-a}+1\right)}.
\end{equation*}

For notational brevity, we will denote the function $p_{G}(\expyg,\expyr,\epsilon_{R})$ by $p_{G}$, when the parameters in question are clear. We have that\todo{Try to align so that all the ifs line up as well}  

	\begin{equation}\label{eq:expec_pi_1}
		\mathbb{E}_{\epsilon_{R}}[u_{G}(\beta,\epsilon_{R})]=
		\begin{cases}
			&\alpha\tilde{p}_{G} \text{ if } \beta \in \left[0,{\beta_{G}}^{-}\right)\\
			& \alpha\mathbb{E}_{\epsilon_{R}}[p_{G}|t_{G}<\epsilon_{R}]\mathbb{P}(t_{G}<\epsilon_{R}) + \beta\mathbb{P}(t_{G}>\epsilon_{R}) \text{ if } \beta \in \left[{\beta_{G}}^{-},{\beta_{G}}^{+}\right)\\
			& \beta \text{ if } \beta \in [{\beta_{G}}^{+},{\beta_{R}}^{-}]\\
			& \alpha\mathbb{E}_{\epsilon_{R}}[p_{G}|t_{R}>\epsilon_{R}]\mathbb{P}(t_{R}>\epsilon_{R}) + \beta\mathbb{P}(t_{R}<\epsilon_{R}) \text{ if } \beta \in \left({\beta_{R}}^{-},{\beta_{R}}^{+}\right]\\
			&\alpha\tilde{p}_{G} \text{ if } \beta \in \left({\beta_{R}}^{+},1\right].
		\end{cases} 
	\end{equation}
\subsubsection{Proof for \Cref{prop:beta_unique}}
Before we commence with our proof, we need the following lemmas:

\begin{lemma}\label{lem:piece-wise}
	$\mathbb{E}_{\epsilon_{R}}[u_{G}(\beta,\epsilon_{R})]$ is piece-wise continuous in $\beta$. 
\end{lemma}

\begin{proof}
	It is clear that the function in \ref{eq:expec_pi_1} is continuous on the interior of each segment on which it is defined. We therefore need only show that is it continuous at the boundary of each interval.
	
	\underline{At ${\beta_{G}}^{-}$:} here we need to show that
	\begin{align*}
		\alpha p_{G} = \alpha\expec{\epsilon_{R}}{p_{G}|t_{G}({\beta_{G}}^{-})<\epsilon_{R}}\prob{t_{G}({\beta_{G}}^{-})<\epsilon_{R}} + {\beta_{G}}^{-}\prob{t_{G}({\beta_{G}}^{-})>\epsilon_{R}}.
	\end{align*}
	where $t_{1}({\beta_{G}}^{-})$ is the value of the upper bound on the shock for the government to be willing to declare war.  
	
	Indeed, we have that 
	\begin{align*}
		&\alpha\mathbb{E}_{\epsilon_{R}}[p_{G}|t_{G}({\beta_{G}}^{-})<\epsilon_{R}]\mathbb{P}(t_{G}({\beta_{G}}^{-})<\epsilon_{R}) + {\beta_{G}}^{-}\mathbb{P}(t_{G}({\beta_{G}}^{-})>\epsilon_{R}) \\
		=& \frac{\alpha}{2a} \int_{-a}^{t_{G}({\beta_{G}}^{-})}\frac{
			e^{\expyg}}{e^{\expyg}+e^{\expyr+\epsilon_{R}}}\mathrm{d}\epsilon_{R}  + {\beta_{G}}^{-} \left(1-\frac{t_{G}({\beta_{G}}^{-})+a}{2a}\right)\\
		=& \frac{\alpha}{2a}\left[a+\expyg-\expyr + \log\left(\frac{\alpha}{{\beta_{G}}^{-}}-1\right) -\log\left(e^{\expyg-\expyr}\frac{\alpha}{{\beta_{G}}^{-}}\right) -a + \log\left(e^{\expyg-\expyr+a}+1\right)\right] \\
		&+ {\beta_{G}}^{-} \left(1-\frac{\expyg-\expyr+a+\log\left(\frac{\alpha}{{\beta_{G}}^{-}}-1\right)}{2a}\right).
	\end{align*}
	
	Here, it is helpful to note that 
	\begin{equation*}
		\frac{\alpha}{{\beta_{G}}^{-}} = e^{\expyr-\expyg+a}+1.
	\end{equation*}
	
	Thus, the above simplifies to
	\begin{align*}
		&{\beta_{G}}^{-} \left(1-\frac{\expyg-\expyr+a+\expyr-\expyg+a}{2a}\right) + \frac{\alpha}{2a}\left[2a - \log\left(e^{\expyg-\expyr}\left(e^{\expyg-\expyr+a}+1\right)\right) -a+log\left(e^{\expyg-\expyr+a}+1\right)\right]\\
		&=\frac{\alpha}{2a}\log\left(\frac{e^{\expyg-\expyr+a}+1}{e^{\expyg-\expyr-a}+1}\right)
	\end{align*}
	
	\underline{At ${\beta_{G}}^{+}$:} similarly to the above, we need to show that 
	\begin{equation*}
		{\beta_{G}}^{+} = \alpha\expec{\epsilon_{R}}{p_{G}|t_{G}({\beta_{G}}^{+})<\epsilon_{R}}\prob{t_{G}({\beta_{G}}^{+})<\epsilon_{R}} + {\beta_{G}}^{+}\prob{t_{G}({\beta_{G}}^{+})>\epsilon_{R}}.
	\end{equation*}
	Indeed, we have that:
	\begin{align*}
		& \frac{\alpha}{2a}\left[a+\expyg-\expyr + \log\left(\frac{\alpha}{{\beta_{G}}^{+}}-1\right) -\log\left(e^{\expyg-\expyr}\frac{\alpha}{{\beta_{G}}^{+}}\right) -a + \log\left(e^{\expyg-\expyr+a}+1\right)\right] \\
		&+ {\beta_{G}}^{+} \left(1-\frac{\expyg-\expyr+a+\log\left(\frac{\alpha}{{\beta_{G}}^{+}}-1\right)}{2a}\right).
	\end{align*}
	It is helpful to note that:
	\begin{equation}
		\frac{\alpha}{{\beta_{G}}^{+}} = e^{\expyr-\expyg-a}+1.
	\end{equation}
	Thus, the above simplifies to
	\begin{align*}
		&\frac{\alpha}{2a}\left[a+\expyg-\expyr+\expyr-\expyg-a-\log\left(e^{\expyg-\expyg}\left(1+e^{\expyr-\expyg-a}\right)\right)-a+log\left(e^{\expyg-\expyr+a}+1\right)\right]\\
		& + {\beta_{G}}^{+}\left(1-\frac{\expyg-\expyr-a+\expyr-\expyg+a}{2a}\right)\\
		&= {\beta_{G}}^{+},
	\end{align*}
	as desired. 
	
	At ${\beta_{R}}^{-}$: here, we need to show that 
	\begin{equation*}
		{\beta_{R}}^{-} = \alpha\expec{\epsilon_{R}}{p_{G}|t_{R}({\beta_{R}}^{-})>\epsilon_{R}}\prob{t_{R}({\beta_{R}}^{-})>\epsilon_{R}} + {\beta_{R}}^{-}\prob{t_{R}({\beta_{R}}^{-})<\epsilon_{R}}.
	\end{equation*}
	
	Indeed, we have that:
	\begin{align*}
		&\alpha\expec{\epsilon_{R}}{p_{G}|t_{R}({\beta_{R}}^{-})>\epsilon_{R}}\prob{t_{R}({\beta_{R}}^{-})>\epsilon_{R}} + {\beta_{R}}^{-}\prob{t_{R}({\beta_{R}}^{-})<\epsilon_{R}}\\
		&= \frac{\alpha}{2a}\int_{t_{2}(({\beta_{R}}^{-}))}^{a}\frac{
			e^{\expyg}}{e^{\expyg}+e^{\expyr+\epsilon_{R}}}\mathrm{d}\epsilon_{R}+{\beta_{R}}^{-}\left(\frac{t_{R}({\beta_{R}}^{-})+a}{2a}\right)\\
		&=\frac{\alpha}{2a}\left[\log\left(1-{\beta_{R}}^{-}\right) + \log\left(\frac{1}{1-{\beta_{R}}^{-}}\right)\right] +\frac{{\beta_{R}}^{-}}{2a}\left[a+\expyg-\expyr-\log\left(\frac{1-\alpha-{\beta_{R}}^{-}}{{\beta_{R}}^{-}-1}\right)\right].
	\end{align*}
	Here, it is helpful to note that:
	\begin{equation*}
		\frac{1-\alpha-{\beta_{R}}^{-}}{{\beta_{R}}^{-}-1} = e^{\expyg-\expyr-a}.
	\end{equation*}
	Thus, the above simplifies to
	\begin{align*}
		&= \frac{{\beta_{R}}^{-}}{2a}\left(a+\expyg-\expyr-\expyg+\expyr+a\right)\\
		&={\beta_{R}}^{-},
	\end{align*}
	as desired.
	
	\underline{At ${\beta_{R}}^{+}$:} here we wish to show that
	\begin{equation*}
		\alpha p_{G} = \frac{\alpha}{2a}\left[\log\left(1-{\beta_{R}}^{-}\right) + \log\left(\frac{1}{1-{\beta_{R}}^{+}}\right)\right] +\frac{{\beta_{R}}^{+}}{2a}\left[a+\expyg-\expyr-\log\left(\frac{1-\alpha-{\beta_{R}}^{+}}{{\beta_{R}}^{+}-1}\right)\right].
	\end{equation*}
	Here, it is useful to note that 
	\begin{equation*}
		\frac{1-\alpha-{\beta_{R}}^{+}}{{\beta_{R}}^{+}-1} = e^{\expyg-\expyr+a}.
	\end{equation*}
	Thus the above simplifies to
	\begin{align*}
		&\frac{\alpha}{2a}\left[\log(\alpha) - \log\left(1+e^{\expyg-\expyr-a}\right) - \log(\alpha) + \log\left(1+e^{\expyg-\expyr+a}\right)\right] + \frac{{\beta_{R}}^{+}}{2a}\left[a+\expyg-\expyr-(a+\expyg-\expyr)\right]\\
		&=\alpha p_{G},
	\end{align*}
	and thus completes the proof.
\end{proof}

\begin{lemma}\label{lem:pos_deriv}
	The derivative of $\expec{\epsilon_{R}}{u_{G}(\beta,\epsilon_{R})}$ with respect to $\beta$ is positive on $\left({\beta_{G}}^{-}, {\beta_{G}}^{+}\right]$.
\end{lemma}

\begin{proof}
	On $\left[{\beta_{G}}^{-}, {\beta_{G}}^{+}\right]$, the derivative is given by:
	\begin{align}
		\nonumber \dfrac{\partial}{\partial\beta}\expec{\epsilon_{R}}{u_{G}(\beta,\epsilon_{R})} &= \nonumber\dfrac{\partial}{\partial\beta}\frac{\beta}{2a}\left[a - y_{1} + y_{2} - \log{\left(\frac{\alpha}{\beta} - 1 \right)}\right] \\&+ \nonumber \frac{\alpha}{2a}\left[\log\left(\frac{\alpha}{\beta}-1\right) -\log\left(\frac{\alpha}{\beta}\right)+\log{\left(e^{\expyg-\expyr+a}+1\right)}\right]\\
		&= \frac{1}{2a}\left(a+\expyr-\expyg + \log\left(\frac{\beta}{\alpha-\beta}\right)\right)\label{eq:2_no_fight_diff}.
	\end{align}
	Note that the logarithm is well defined as $\alpha>{\beta_{G}}^{+}$.  We have that \ref{eq:2_no_fight_diff} is strictly positive iff
	\begin{align*}
		\log\left(\frac{\beta}{\alpha-\beta}\right) &> \expyg-\expyr-a\\
		\iff \frac{\alpha}{\beta}-1 &< e^{\expyr-\expyg+a}\\
		\iff \frac{1}{\beta} &< \frac{e^{\expyg-\expyg+a}+1}{\alpha}\\
		\iff \beta &> \frac{\alpha}{e^{\expyg-\expyr+1}+1} = {\beta_{G}}^{-}.
	\end{align*}
	Thus the result holds.
\end{proof}

\begin{lemma}\label{lem:neg_deriv}
	Suppose $\alpha<1$. The derivative of $\mathbb{E}_{\epsilon_{R}}[u_{G}(\beta,\epsilon_{R})]$ with respect to $\beta$ is strictly negative at ${\beta_{R}}^{+}.$
\end{lemma}

\begin{proof}
	On $\left({\beta_{R}}^{-},{\beta_{R}}^{+}\right]$, the derivative of the expected pay-off function with respect to $\beta$ is:
	\begin{equation*}
		\frac{\alpha}{2a(1-\beta)}\left(\frac{\beta}{1-\alpha-\beta}\right) + \frac{1}{2a}\left(a+\expyg-\expyr - \log\left(\frac{1-\alpha-\beta}{\beta-1}\right)\right).
	\end{equation*}
	Thus, at ${\beta_{R}}^{+}$, it evaluates to
	\begin{align*}
		&\frac{\alpha}{2a(1-{\beta_{R}}^{+})}\left(\frac{\alpha-1}{\alpha}\left(1+e^{\expyr-\expyg+a}\right)\right) + \frac{1}{2a}\left(a+\expyg-\expyr-\log\left(e^{\expyg-\expyr+1}\right)\right)\\
		&= \left(\alpha-1\right)\frac{1+e^{\expyr-\expyg+a}}{2a(1-{\beta_{R}}^{+})}.
	\end{align*} 
	Since $(\alpha-1)<0$ $\forall \alpha$, and $\frac{1+e^{\expyr-\expyg+a}}{2a(1-\beta)}>0$ $\forall \beta$, the result holds.
\end{proof}

\begin{prop}\label{prop:beta_interval}
    $\beta^{*}$ lies in $\left[\beta_{R}^{-},\beta_{R}^{+}\right)$.
\end{prop}

\begin{proof}
	We may prove the above statement by finding the optimal choice of $\beta$ on each continuous segment of the piecewise function in \Cref{eq:expec_pi_1}. We may then compare the payoff at each optimal $\beta$ on each segment and determine the global optimal value(s), $\beta^{*}$. Because, by \Cref{lem:piece-wise}, the function is piece-wise continuous we may consider closed intervals when finding $\beta^{*}$. This makes solving the optimisation problem easier. 
	
	\Cref{lem:pos_deriv} combined with the fact that payoff is constant on $[0,{\beta_{G}}^{-}]$ (as war occurs with certainty), and linearly increasing on $[{\beta_{G}}^{+},{\beta_{R}}^{-}]$ (as both parties will accept all values in this interval), means that $\mathbb{E}_{\epsilon_{R}}[u_{G}(\beta,\epsilon_{R})]$ is weakly increasing on the entire interval $[0,{\beta_{R}}^{-}]$. Thus $\beta^{*}$ will never be less than ${\beta_{R}}^{-}$.
	
	We know from \Cref{lem:neg_deriv} that the optimum will never occur on the upper bound of $\left[\beta_{R}^{-},\beta_{R}^{+}\right]$ as the derivative here is negative. Moreover, the government is indifferent between all $\beta \in [{\beta_{R}}^{+},1]$ as war once again occurs with certainty here. Because  $\mathbb{E}_{\epsilon_{R}}[u_{G}(\beta,\epsilon_{R})]$ is piecewise continuous, the fact that its function is negative at ${\beta_{R}}^{+}$ means that there is some $\beta<{\beta_{R}}^{+}$ which is preferred to $\beta \in [{\beta_{R}}^{+},1]$. We therefore have the desired result. 
\end{proof}

\begin{lemma}\label{lem:interior_unique}
	If $\beta^{*}$ is interior (i.e. in $\left({\beta_{R}}^{-},{\beta_{R}}^{+}\right)$), then its value is unique.
\end{lemma}

\begin{proof}
	We have that the first order condition for $\beta^{*}$ being a local maximum on the interval $\left[\betaint\right]$ is given by 
	\begin{equation}\label{eq:FOC}
		\expyg-\expyr+a + \frac{\alpha(1-\alpha)}{(1-\beta)(1-\beta -\alpha)} -\log\left(\frac{1-\beta-\alpha}{\beta-1}\right)=0. 
	\end{equation}
	It is sufficient to prove that \Cref{eq:FOC} has at most two solutions, as if a function has two stationary points then one must be a local maximum and the other a local minimum. We shall therefore prove that the first derivative of \Cref{eq:FOC} with respect to $\beta$ has exactly one root and thus \Cref{eq:FOC} may cross the axis at most twice. Indeed we have that the first derivative is 
	\begin{equation*}
		\frac{\alpha \left(\alpha^{2} -4 \alpha + 3 + (3 \alpha -4)\beta  + \beta^{2} \right)}{\left(\beta - 1\right)^{2} \left(\alpha + \beta - 1\right)^{2}}.
	\end{equation*}
	This has two roots given by
	\begin{equation*}
		\beta_{\pm} = \frac{1}{2} (4-3 \alpha )\pm \frac{1}{2} \sqrt{5 \alpha ^2-8 \alpha +4}.
	\end{equation*}
	However, this is a decreasing function of $\alpha$ on $[0,1]$, with minimum value $3/2$ at $\alpha=1$. Since $\beta\leq 1$ this is not a valid solution for any $\alpha$. 
\end{proof}

\begin{prop}\label{prop:beta_deriv_a}
    If $\tilde{\beta}$ is a local maximum of $\mathbb{E}_{\epsilon_{R}}[u_{G}(\beta,\epsilon_{R})]$ on $[{\beta_{R}}^{-},{\beta_{G}}^{-}]$, then it is increasing in $a$.
\end{prop}

\begin{proof}
	We can prove this by differentiating the first order condition for $\beta$ of $\mathbb{E}_{\epsilon_{R}}[u_{G}(\beta,\epsilon_{R})]$ on $[{\beta_{R}}^{-},{\beta_{G}}^{-}]$ with respect to $a$. 
	
	At $\beta^{*}$, we have that
	\begin{align*}
		&\frac{\alpha(1-\alpha)}{(1-\beta^{*})(1-\alpha -\beta^{*})} + \log\left(\frac{\beta^{*}-1}{1-\alpha-\beta^{*}}\right) = \expyg-\expyr-a.\\
		&\Rightarrow \dfrac{\partial}{\partial\beta^{*}}\left(\frac{\alpha(1-\alpha)}{(1-\beta^{*})(1-\alpha -\beta^{*})} + \log\left(\frac{\beta^{*}-1}{1-\alpha-\beta^{*}}\right)\right) \dfrac{\partial \beta^{*}}{\partial a} = -1\\.
	\end{align*}
	But notice that the derivative with respect to $\beta^{*}$ in the line above is, in fact, the second derivative of $\mathbb{E}_{\epsilon_{R}}[u_{G}(\beta,\epsilon_{R})]$ on $[{\beta_{R}}^{-},{\beta_{G}}^{-}]$, evaluated at $\beta_{*}$. For this to be a local maximum, the second derivative must be negative at $\beta^{*}$. Thus $\dfrac{\partial \beta^{*}}{\partial a}$  must be positive. 
\end{proof}

\begin{prop}\label{prop:beta_val_unique}
      If $a<\expyg-\expyr$, or if $\alpha$ is sufficiently large relative to $\expyg-\expyr-a$, then the equilibrium value for $\beta^{*}$ is unique for every value of $a$. Otherwise, the equilibrium is unique for all but a single value of $a$. 
\end{prop}

\begin{proof}
     We begin by noticing that $\beta_{R}^{-}$ satisfies the first order conditions on $\mathbb{E}_{\epsilon_{R}}[u_{G}(\beta,\epsilon_{R})]$ for at most one value of $a$. From \Cref{cor:beta_vals} we know that ${\beta_{R}}^{-}$ is decreasing in $a$. We know from \Cref{prop:beta_deriv_a} we know that if $\tilde{beta}$ is a local maximum then it is monotonically increasing in $a$. Thus, we have that $\tilde{\beta}-{\beta_{R}}^{-}$ is a monotonically increasing function in $a$ and has at most one root.

	From \Cref{prop:beta_interval}, we know that $\beta_{*}$ must lie in $\left[\betaint\right]$, where we may consider the closed interval without loss of generality as our function is piecewise continuous. 
	
	We know from \Cref{lem:interior_unique} that any local maximum value on this interval if it exists, is unique. Moreover, we know that $\beta_{R}^{+}$ is strictly dominated in some open neighbourhood $\left(\tilde{\beta},{\beta_{R}}^{+}\right)$ in the interval as expected payoff is strictly negative at ${\beta_{R}}^{+}$.
	
	Thus the choice of $\beta^{*}$ will be unique if the expected payoff at the local maximum is never equal to the expected payoff at ${\beta_{R}}^{-}$, unless the local maximum occurs at the boundary. To do so we show that if $\beta^{*}$ is a local maximum, then 
	\begin{equation*}
		\mathbb{E}[u_{G}(\beta^{*}, \epsilon_{R})] - {\beta_{R}}^{-}  
	\end{equation*}
	has exactly one root, precisely when $\beta^{*} = {\beta_{R}}^{-}$. Indeed, we begin by noticing that for any $\beta^{*}$ on the interval $\left[\betaint\right]$ we have that the payoff for the government is given by
	\begin{align*}
		\mathbb{E}[u_{G}(\beta^{*})] &= \frac{\alpha}{2a}\left(\log\left(\frac{1-{\beta_{R}}^{-}}{1-\beta^{*}}\right)\right) + \frac{\beta^{*}}{2a}\left(a+\expyg-\expyr - \log\left(\frac{1-\alpha-\beta^{*}}{\beta^{*}-1}\right)\right).
	\end{align*}
	We may use the first order conditions on $\beta^{*}$, 
	\begin{equation}
		\frac{\partial \mathbb{E}[u_{G}(\beta)]}{\partial \beta}\bigg|_{\beta^{*}} =\frac{1}{2a}\left(\frac{\alpha(1-\alpha)}{(1-\beta^{*})(1-\alpha-\beta^{*})} + a + \expyg-\expyr - \log\left(\frac{1-\alpha - \beta^{*}}{\beta^{*}-1}\right) \right) = 0. 
	\end{equation}
	to simplify the second term so that
	\begin{equation}
		\mathbb{E}[u_{G}(\beta)] = \frac{\alpha}{2a}\left(\log\left(\frac{1-{\beta_{R}}^{-}}{1-\beta^{*}}\right)\right) - \frac{\beta^{*}}{2a}\left(\frac{\alpha(1-\alpha)}{(1-\beta^{*})(1-\alpha-\beta^{*})}\right).
	\end{equation}
	Thus, we must have that	
	\begin{align*}
		{\beta_{R}}^{-} &= \frac{\alpha}{2{a}}\left(\log\left(\frac{1-{\beta_{R}}^{-}}{1-\beta^{*}}\right)\right) - \frac{\beta^{*}}{2{a}}\left(\frac{\alpha(1-\alpha)}{(1-\beta^{*})(1-\alpha-\beta^{*})}\right).
	\end{align*}

	Rearranging the above gives us
	\begin{equation*}
		\alpha \log\left(\frac{1}{1-\beta^{*}}\right) -\beta^{*}\frac{\alpha(1-\alpha)}{(1-\beta^{*})(1-\alpha-\beta^{*})} = 2a {\beta_{R}}^{-} - \alpha\log\left(\frac{1}{1-{\beta_{R}}^{-}}\right)  .
	\end{equation*}
	We now wish to show that for any given value of ${\beta_{R}}^{-}$, there is at most one value of $\beta^{*}$ which satisfies this equation. We do so by showing that the left-hand side of the equation is monotonic in $\beta^{*}$. Indeed, we have that 
	\begin{equation*}
		-\alpha\dfrac{\partial}{\partial \beta^{*}}\log(1-\beta^{*})
		-\frac{\partial}{\partial\beta^{*}}\left(\frac{\beta^{*}\alpha(1-\alpha)}{(1-\beta^{*})(1-\alpha-\beta^{*})}\right) = \frac{\alpha(1-\beta^{*})(\alpha +\beta^{*} -1)^2 + (1-\alpha) \alpha  \left(\alpha +(\beta^{*})^2 - 1\right)}{(\beta^{*} -1)^2 (\alpha +\beta^{*} -1)^2}.
	\end{equation*}
	
	It is clear that this is positive when $\left(\alpha +(\beta^{*})^2-1\right)>0$, which occurs when $\beta^{*} > \sqrt{1-\alpha}$. However, recall that we are studying $\beta^{*} \in \left[\betaint\right]$. Thus it is sufficient to prove that ${\beta_{R}}^{-} > \sqrt{1-\alpha}$:
	\begin{align*}
		1-\frac{\alpha}{1+e^{\expyg-\expyr-a}} &> \sqrt{1-\alpha}\\
		\iff \frac{2\alpha}{1+e^{\expyg-\expyr-a}} - \frac{\alpha^2}{\left(1+e^{\expyg-\expyr-a}\right)^{2}} &< \alpha\\
		\iff - \left(1+e^{\expyg-\expyr-a}\right)^{2} + 2 \left(1+e^{\expyg-\expyr-a}\right) -\alpha &< 0 \\
		\iff (e^{\expyg-\expyr-a}+\sqrt{1-\alpha})(e^{\expyg-\expyr-a}-\sqrt{1-\alpha}) > 0.
	\end{align*}
	The above inequality is only satisfied if $e^{\expyg-\expyr-a}<-\sqrt{1-\alpha}$ (which is never possible) or $e^{\expyg-\expyr-a}>\sqrt{1-\alpha}$. However, this is certainly true if $\expyg-\expyr-a>0$, but also true if $1-e^{2*(\expyg-\expyr-a)} < \alpha$.
\end{proof}

\Cref{prop:beta_unique} follows immediately from \Cref{prop:beta_interval} and \Cref{prop:beta_val_unique}.

\subsubsection{Proof of \Cref{thm:non_mono}}
\begin{proof}
        \emph{Strategy switch in increasing $a$:} In the proof of \Cref{prop:beta_val_unique}, we show that if $\tilde{\beta}$ is an interior equilibrium, then $\beta_{R}^{-} - \mathbb{E}_{\epsilon_{R}}[u_{G}(\tilde{\beta},\epsilon_{R})]$ is decreasing in $\tilde{\beta}$. Because $\tilde{\beta}$ is increasing in $a$, it remains to show that $\mathbb{E}_{\epsilon_{R}}[u_{G}(\tilde{\beta},\epsilon_{R})] > \beta_{R}^{-}$ for some lower bound on $a$.
        
	To prove this, it is sufficient to show that 
	\begin{equation*}
		\dfrac{\partial}{\partial \beta}\mathbb{E}_{\epsilon_{R}}[u_{G}(\beta,\epsilon_{R})]|_{{\beta_{R}}^{-}} >0;
	\end{equation*}
	for at least some of the parameter space.
	
	We have that 
	\begin{align*}
		&\dfrac{\partial}{\partial \beta}\mathbb{E}_{\epsilon_{R}}[u_{G}(\beta,\epsilon_{R})]|_{{\beta_{R}}^{-}} \\
		&= \left(\frac{\alpha}{1-{\beta_{R}}^{-}}\right)\left(\frac{1-\alpha}{1-\alpha-{\beta_{R}}^{-}}\right) + a+\expyg-\expyr + \log\left(\frac{{\beta_{R}}^{-}-1}{1-\alpha-{\beta_{R}}^{-}}\right).
	\end{align*}
	
	We may simplify the above by noticing that
	\begin{align*}
		(1-{\beta_{R}}^{-})(1-\alpha-{\beta_{R}}^{-}) &=-\left(\frac{\alpha}{1+e^{\expyg-\expyr-a}}\right)\left(\frac{\alpha}{1+e^{\expyr-\expyg+a}}\right)\\
		&=-\frac{\alpha^{2}e^{\expyr-\expyg-a}}{\left(1+e^{\expyg-\expyr-a}\right)^{2}}.
	\end{align*}
	
	Similarly:
	\begin{align*}
		\frac{{\beta_{R}}^{-}-1}{1-{\beta_{R}}^{-}-\alpha} &= \left(-\frac{\alpha}{1+e^{\expyg-\expyr-a}}\right)\left(-\frac{1+e^{\expyr-\expyg+a}}{\alpha}\right)
		&= e^{\expyr-\expyg+a}.
	\end{align*}
	
	Thus, we have that:
	\begin{align*}
		&\dfrac{\partial}{\partial \beta}\mathbb{E}_{\epsilon_{R}}[u_{G}(\beta,\epsilon_{R})]|_{{\beta_{R}}^{-}}\\
		&= \frac{\alpha (1-\alpha)}{-\alpha^{2}}\frac{\left(1+e^{\expyg-\expyr-a}\right)^{2}}{e^{\expyr-\expyg-a}} + a+\expyg-\expyr+\log\left(e^{\expyr-\expyg+a}\right)\\
		&= \frac{(1-\alpha)}{-\alpha}\frac{\left(1+e^{\expyg-\expyr-a}\right)^{2}}{e^{\expyr-\expyg-a}} + 2a.
	\end{align*}
	
	And so the derivative is positive at the boundary if
	\begin{align*}
		2a>\frac{(1-\alpha)}{\alpha}\frac{\left(1+e^{\expyg-\expyr-a}\right)^{2}}{e^{\expyr-\expyg-a}}.
	\end{align*} 
	
	To show that this condition is satisfied, we need to find the values of the parameter space for which this is true.  
	
	To this end, let $x = e^{\expyg-\expyr-a}$. We wish to solve for the values of $x$ such that
	\begin{align*}
		{2a}&>\frac{(1-\alpha)}{\alpha}\frac{\left(1+x\right)^{2}}{x}\\
		\iff 0 &> (1-\alpha)x^{2}+\left(1-2\alpha(1-a)\right)x +(1-\alpha).
	\end{align*}
	The roots of the polynomial on the RHS of the inequality are given by
	\begin{align*}
		x_{+,-} &= \frac{\pm \sqrt{a\alpha}\sqrt{\alpha(a+2) -2} + (a\alpha +\alpha -1)}{1-\alpha}\\ 
		&= \frac{\pm \sqrt{(\alpha(a+1)-1)^{2}-(\alpha^{2}+1)}+ (\alpha(a+1)-1)}{1-\alpha}.
	\end{align*}
	
	For any real roots to exist, we need that 
	\begin{equation*}
		\alpha(a+2)\geq 2 \iff \alpha \geq \frac{2}{a+2}. 
	\end{equation*}
	
	Notice that
	\begin{equation}\label{eq:real_root}
		\frac{2}{a+2} \geq \frac{1}{a+1},
	\end{equation}
	if $a>0$.
	
	The polynomial is concave up. So, if  $\alpha(a+1)<1$ then there are no positive roots- and thus the function is positive $\forall x\geq 0$. But $\alpha(a+1)<1$ is impossible by \ref{eq:real_root}. If $\alpha(a+1)>1$ then both roots are positive: so the function is negative on $x_{-}<x<x_{+}$. Taking a logarithm across this inequality gives the desired upper and lower bounds:
	\begin{align*}
		\log\left(-\frac{\sqrt{\eta^{2}-(\alpha^{2}+1)} + \eta}{1-\alpha}\right)+ a \leq (\expyg-\expyr) \leq \log\left(\frac{\sqrt{\eta^{2}-(\alpha^{2}+1)} + \eta}{1-\alpha}\right)+a
	\end{align*}
	Where $\eta = (\alpha(a+1)-1)$.		
    \emph{Non-monotonicity:} this follows immediately from \cref{cor:beta_vals} and \cref{prop:beta_deriv_a}.
\end{proof}

\subsubsection{Proof of \Cref{prop:prob_war}}
\begin{proof}
	We have that the probability of war, in terms of $\beta^{*}$, is given by the probability that $R$ wants to fight:
	\begin{equation*}
		\frac{1}{2a} \left(a-y_{G}+y_{R} + \log\left(\frac{\alpha}{1-\beta^{*}}-1\right)\right).
	\end{equation*}
	Differentiating the above with respect to $a$ gives
	\begin{equation}
		-\frac{1}{2a^{2}}\left(a -y_{G} + y_{R} + log\left(\frac{\alpha}{1-\beta}-1\right)\right) + \frac{1}{2a}\left(1+ \frac{\alpha}{(1-\beta^{*})(\alpha + \beta^{*} -1)}\dfrac{\partial \beta^{*}}{\partial a}\right).
	\end{equation}
	The first term is positive as 
	\begin{align}
		log\left(\frac{\alpha}{1-\beta^{*}}-1\right) \geq  log\left(\frac{\alpha}{1-\beta_{R}^{-}}-1\right) = y_{G}-y_{R}-a,
	\end{align}
	where the inequality follows from the fact that $\beta^{*} \geq \beta_{R}^{-}$.
	
	The second term is positive as ${\partial \beta^{*}}/{\partial a}>0$ by \Cref{prop:beta_deriv_a} and we also have that ${\alpha}/((1-\beta^{*})(\alpha + \beta^{*} -1)) \geq 0$ as
	\begin{align*}
		 \alpha + \beta^{*} \geq \alpha + \beta_{R}^{-} = 1 +\alpha \frac{\alpha}{1+e^{y_R-y_G-a}} >1,
	\end{align*}
	so that $\alpha + \beta^{*} -1 >0$.
\end{proof}

\subsubsection{Proof of \Cref{prop:welfare_decreasing}}

\begin{proof}
	The proof for when $a \leq a_{crit}$ is trivial as no fighting occurs and total welfare is equal to $1$. When $a>a_{crit}$ then we have that total welfare is given by $1-\mathbb{P}(\epsilon_{R}>t_{R}(\beta^{*}))(1-\alpha)$. Since \Cref{prop:prob_war} shows that $\mathbb{P}(\epsilon_{R}>t_{R}(\beta^{*}))$ is increasing in $a$, the result follows immediately.
\end{proof}

        \bibliographystyle{agsm}
        \bibliography{Bibliography}
    	
\end{document}